%% file: main.tex
\documentclass[11pt]{article}
\usepackage{fullpage}
\usepackage{amsmath,amsfonts,amsthm,amssymb}
\usepackage{url}
\usepackage[usenames,dvipsnames,svgnames,table]{xcolor}
\usepackage[pdftex,bookmarks,colorlinks=true,
            linkcolor=red,
            urlcolor=blue,
            citecolor=gray]{hyperref}
\usepackage{color}
\usepackage{algorithm,algorithmic}
\usepackage{dsfont}
\newtheorem{theorem}{Theorem}
\newtheorem{lemma}[theorem]{Lemma}
\newtheorem{corollary}[theorem]{Corollary}

\newtheorem{definition}[theorem]{Definition}
\newcommand{\specialcell}[2][c]{%
  \begin{tabular}[#1]{@{}c@{}}#2\end{tabular}}

\input{bib_macros} 

\input{tex_macros} 
\include{general_notation} 
\include{praneeth_notations}

\begin{document}

\title{Robust Shift-and-Invert Preconditioning: Faster and More Sample Efficient Algorithms for Eigenvector Computation~\footnote{The manuscript is out of date. An updated version of this work is available at \cite{newVersion}.}}
\date{November 2, 2015}

\author{Chi Jin \\ UC Berkeley \\ \texttt{chijin@eecs.berkeley.edu}
\and
Sham M. Kakade \\ University of Washington \\ \texttt{sham@cs.washington.edu}
\and
Cameron Musco\\MIT\\ \texttt{cnmusco@mit.edu}
\and
Praneeth Netrapalli \\ Microsoft Research, New England \\ \texttt{praneeth@microsoft.com}
\and
Aaron Sidford \\ Microsoft Research, New England \\ \texttt{asid@microsoft.com} }
\maketitle

\vspace{-1.3em}
\input{abstract.tex}

\pagebreak
\input{introduction.tex}
\input{preliminaries.tex}

\input{framework.tex}
\input{offline.tex}
\input{online.tex}

\input{paramfree.tex}

\input{lower.tex}

\section{Acknowledgements}
We would like to thank Dan Garber and Elad Hazan for helpful discussions on our work.

\bibliographystyle{alpha}
\bibliography{refs}

\appendix

\section{Appendix}

\begin{lemma}[Eigenvector Estimation via Spectral Norm Matrix Approximation]\label{spectral_error_conversion}
Let $\bv{A}^\top\bv{A}$ have top eigenvector $1$, top eigenvector $v_1$ and eigenvalue gap $\gap$. Let $\bv{B}^\top \bv{B}$ be some matrix with $\norm{\bv{A}^\top\bv{A}-\bv{B}^\top\bv{B}}_2 \le O()\sqrt{\epsilon} \cdot \gap)$. Let $x$ be the top eigenvector of $\bv{B}^\top \bv{B}$. Then:
\begin{align*}
|x^\top v_1 | \ge 1 - \epsilon ~.
\end{align*}
\end{lemma}
\begin{proof}
We can any unit vector $y$ as $y = c_1 v_1 + c_2 v_2$ where $v_2$ is the component of $x$ orthogonal to $v_1$ and $c_1^2 + c_2^2 = 1$. We know that
\begin{align*}
v_1^\top \bv{B}^\top \bv{B} v_1 &= v_1^\top \bv{A}^\top\bv{A}v_1-v_1^T(\bv{A}^\top\bv{A}-\bv{B}^\top\bv{B})v_1\\
1 -\sqrt{\epsilon} \gap &\le v_1^\top \bv{B}^\top \bv{B} v_1\le 1 +\sqrt{\epsilon} \gap
\end{align*}
Similarly we can compute:
\begin{align*}
v_2^\top \bv{B}^\top \bv{B}  v_2 &= v_2^\top \bv{A}^\top\bv{A}v_2-v_2^T(\bv{A}^\top\bv{A}-\bv{B}^\top\bv{B})v_2\\
1 -\gap -\sqrt{\epsilon} \gap &\le v_2^\top \bv{B}^\top \bv{B} v_2 \le 1 -\gap +\sqrt{\epsilon} \gap.
\end{align*}
and 
\begin{align*}
|v_1^\top \bv{B}^\top \bv{B}  v_2| &= |v_1^\top \bv{A}^\top\bv{A}v_2-v_1^T(\bv{A}^\top\bv{A}-\bv{B}^\top\bv{B})v_2|\\
&\le \sqrt{\epsilon} \gap.
\end{align*}

We have $x^\top \bv{B}\bv{B}^\top x = c_1^2 (v_1^\top \bv{B}^\top \bv{B} v_1) + c_2^2 (v_2^\top \bv{B}^\top \bv{B} v_2) + 2c_1c_2 \cdot v_2^\top \bv{B}^\top \bv{B} v_1$. 

We want to bound $c_1 \ge 1-\epsilon$ so $c_1^2 \ge 1- O(\epsilon)$. Since $x$ is the top eigenvector of $\bv{BB}^\top$ we have:
\begin{align*}
x^\top \bv{B}\bv{B}^\top x &\ge v_1^\top \bv{B}\bv{B}^\top v_1\\
c_2^2 (v_2^\top \bv{B}^\top \bv{B} v_2) + 2c_2 v_2^\top \bv{B}^\top \bv{B} v_1 &\ge (1-c_1^2)v_1^\top \bv{B}\bv{B}^\top v_1\\
2\sqrt{1-c_1^2} \sqrt{\epsilon}\gap &\ge (1-c_1^2)\left (v_1^\top \bv{B}\bv{B}^\top v_1 -v_2^\top \bv{B}\bv{B}^\top v_2 \right )\\ 
\frac{1}{\sqrt{1-c_1^2}} &\ge \frac{(1-2\sqrt{\epsilon})\gap}{2\sqrt{\epsilon}\gap}\\
\frac{1}{1-c_1^2} &\ge \frac{1-5\sqrt{\epsilon}}{4\epsilon}
\end{align*}

This means we need have $1-c_1^2 \le O(\epsilon)$  meaning $c_1^2 \ge 1-O(\epsilon)$ as desired.
\end{proof}

\begin{lemma}[Inverted Power Method progress in $\ell_2$ and $\bv{B}$ norms]\label{same_progress}
	Let $x$ be a unit vector with $\inprod{x,v_1} \neq 0$ and let $\xtilde = \mb^{-1}w$, i.e. the power method update of $\mb^{-1}$ on $x$. Then, we have both:
	\begin{align}\label{lb_progress}
\frac{\norm{\bv{P}_{v_1^{\perp}}\xtilde}_{\mb}}{\norm{\bv{P}_{v_1}\xtilde}_{\mb}} \le \frac{\lambda_2(\bv{B}^{-1})}{\lambda_1(\bv{B}^{-1})} \cdot \frac{\norm{\bv{P}_{v_1^{\perp}}x}_{\mb}}{\norm{\bv{P}_{v_1}x}_{\mb}}
	\end{align}
	and 
	\begin{align}\label{l2_progress}
\frac{\norm{\bv{P}_{v_1^{\perp}}\xtilde}_{2}}{\norm{\bv{P}_{v_1}\xtilde}_{2}} \le \frac{\lambda_2(\bv{B}^{-1})}{\lambda_1(\bv{B}^{-1})} \cdot \frac{\norm{\bv{P}_{v_1^{\perp}}x}_{2}}{\norm{\bv{P}_{v_1}x}_{2}}
	\end{align}
\end{lemma}
\begin{proof}
\eqref{lb_progress} was already shown in Lemma \ref{thm:powermethod}. We show \eqref{l2_progress} similarly.

	Writing $x$ in the eigenbasis of $\bv{B}^{-1}$, we have $x=\sum_i \alpha_i v_i$ and  $\xtilde = \sum_i \alpha_i \lamiBinv{i}v_i$. Since $\inprod{x,v_1} \neq 0$, $\alpha_1 \neq 0$ and we have:
	\begin{align*}
		\frac{\norm{\bv{P}_{v_1^{\perp}}\xtilde}_{2}}{\norm{\bv{P}_{v_1}\xtilde}_{2}} = \frac{\sqrt{\sum_{i\geq 2} \alpha_i^2 \lambda^2_{i}(\bv{B}^{-1})}}{\sqrt{\alpha_1^2 \lambda^2_{1}(\bv{B}^{-1})}}
		\leq \frac{\lamiBinv{2}}{\lamiBinv{1}} \cdot \frac{\sqrt{\sum_{i\geq 2} \alpha_i^2}}{\sqrt{\alpha_1^2}}
		= \frac{\lamiBinv{2}}{\lamiBinv{1}} \cdot \frac{\norm{\bv{P}_{v_1^{\perp}}x}_{2}}{\norm{\bv{P}_{v_1}x}_{2}}.
	\end{align*}
\end{proof}
%
%
%
%
%
%

\end{document}

%% file: bib_macros.tex
  \usepackage{nth}
  \usepackage{intcalc}

%% file: tex_macros.tex
\newcommand{\R}{\mathbb{R}}

\newcommand{\E}{\mathbb{E}}
\let\Pr\relax
\DeclareMathOperator*{\Pr}{\mathbb{P}}
\newcommand{\variance}{\mathrm{Var}}

\newcommand{\norm}[1]{\|#1\|}
\newcommand{\normFull}[1]{\left\|#1\right\|}
\newcommand{\normFro}[1]{\norm{#1}_{\mathrm{F}}}

\newcommand{\abs}[1]{\lvert#1\rvert}
\newcommand{\inprod}[1]{\left\langle #1 \right\rangle}

\newcommand{\gradient}{\bigtriangledown}
\newcommand{\grad}{\gradient}
\newcommand{\hessian}{\gradient^{2}}
\newcommand{\hess}{\hessian}

\DeclareMathOperator*{\argmin}{arg\,min}

\DeclareMathOperator{\rank}{rank}
\DeclareMathOperator{\nnz}{nnz}

\newcommand{\bv}[1]{\mathbf{#1}}

\newcommand{\wt}{\widetilde}

\newcommand{\poly}{\mathop\mathrm{poly}}
\newcommand{\otilde}{\widetilde{O}}

\newcommand{\ma}{\mvar{a}}
\newcommand{\mb}{\mvar{b}}

\newcommand{\mSigma}{\mvar{\Sigma}}

\newcommand{\eqdef}{\mathbin{\stackrel{\rm def}{=}}}
\newcommand{\dist}{\mathcal{D}}
\newcommand{\supp}{\mathrm{supp}}

\newcommand{\sk}[1]{\noindent{\textcolor{cyan}{{\bf sk:} \em #1}}}

\newcommand{\gap}{\mathrm{gap}}
\DeclareMathOperator{\nrank}{sr}
\DeclareMathOperator{\nvar}{v}
\newcommand{\rayquot}{\mathrm{quot}}
\newcommand{\ssvrgstep}{\mathrm{ssvrg\_iter}}

\newcommand{\opt}[1]{{#1}^{\mathrm{opt}}}

\makeatletter

%% file: general_notation.tex
\global\long\def\boldVar#1{\mathbf{#1}}
\global\long\def\mvar#1{\boldVar{#1}}

\global\long\def\defeq{\stackrel{\mathrm{{\scriptscriptstyle def}}}{=}}

 \renewcommand{\norm}[1]{\left\|#1\right\|}

\global\long\def\ma{\mvar A}
 \global\long\def\mb{\mvar B}

\global\long\def\mI{\mvar I}

\global\long\def\abs#1{\left|#1\right|}

\global\long\def\ceil#1{\left\lceil #1 \right\rceil }

%% file: praneeth_notations.tex
\global\long\def\lamtilij#1#2{\widetilde{\lambda}_{#2}^{\left(#1\right)}}

\global\long\def\lamhatij#1#2{\widehat{\lambda}_{#2}^{\left(#1\right)}}

\global\long\def\lamj#1{\lambda_{#1}}

\global\long\def\lambari#1{\overline{\lambda}^{\left(#1\right)}}

\newcommand{\expec}[1]{\mathbb{E}\left[#1\right]}
\newcommand{\prob}[1]{\mathbb{P}\left[#1\right]}

\newcommand{\xtilde}{\widetilde{x}}

\newcommand{\lamiBinv}[1]{\lambda_{#1}\left(\mb^{-1}\right)}
\newcommand{\lambdah}{\widehat{\lambda}}

\newcommand{\xhat}{\widehat{x}}
\newcommand{\mcalF}{\mathcal{F}}
\newcommand{\mcalG}{\mathcal{G}}
\newcommand{\rayquoth}[1]{\widehat{\mathrm{quot}}\left(#1\right)}
\newcommand{\solve}[1]{\mathrm{solve}\left(#1\right)}

\newcommand{\set}[1]{\left\{#1\right\}}

\newcommand{\mM}{\mathbf{M}}

%% file: abstract.tex
\begin{abstract}
In this paper we provide faster algorithms and improved sample complexities for approximating the top eigenvector of a matrix $\bv{A}^\top \bv{A}$. In particular we give the following results for computing an $\epsilon$ approximate eigenvector - i.e. some $x$ such that $x^\top \bv{A}^\top \bv{A}x \ge (1-\epsilon)\lambda_1(\bv{A}^\top \bv{A})$:

\begin{itemize}
	\item \textbf{Offline Eigenvector Estimation:}  Given an explicit matrix $\bv{A} \in \R^{n \times d}$, we show how to compute an $\epsilon$ approximate top eigenvector in time $\otilde \left(\left [\nnz(\bv{A}) + \frac{d\nrank(\bv{A})}{\gap^2}\right ]\cdot \log 1/\epsilon \right )$ and $\otilde\left(\left [\frac{\nnz(\bv A)^{3/4} (d\nrank(\bv{A}))^{1/4}}{\sqrt{\gap}}\right ]\cdot \log1/\epsilon \right )$. Here $\nrank(\bv{A})$ is the stable rank, $\gap$ is the multiplicative gap between the largest and second largest eigenvalues, and $\tilde O(\cdot)$ hides log factors in $d$ and $\gap$. By separating the $\gap$ dependence from $\nnz(\bv{A})$ our first runtime improves classic iterative algorithms such as the power and Lanczos methods. It also improves on previous work separating the $\nnz(\bv{A})$ and $\gap$ terms using fast subspace embeddings \cite{ailon2009fast,clarkson2013low} and stochastic optimization \cite{shamir2015stochastic}. We obtain significantly improved dependencies on $\nrank(\bv{A})$ and $\epsilon$ and our second running time improves this further when $\nnz(\bv{A}) \le \frac{d\nrank(\bv{A})}{\gap^2}$.
	
	\item \textbf{Online Eigenvector Estimation:} 
	Given a distribution $\dist$ over vectors $a \in \R^d$ with covariance matrix $\bv{\Sigma}$ and a vector $x_0$ which is an $O(\gap)$ approximate top eigenvector for $\bv{\Sigma}$, we show how to
	 compute an $\epsilon$ approximate eigenvector using $\tilde O \left (\frac{\nvar(\dist)}{\gap^2} + \frac{\nvar(\dist)}{\gap \cdot \epsilon} \right )$ samples from $\dist$.  Here $\nvar(\dist)$ is a natural notion of the variance of $\dist$. 
	Combining our algorithm with a number of existing algorithms to initialize $x_0$ we obtain improved sample complexity and runtime results under a variety of assumptions on $\dist$. Notably, we show that, for general distributions, our sample complexity result is asymptotically optimal - we achieve optimal accuracy as a function of sample size as the number of samples grows large.
\end{itemize}

We achieve our results using a general framework that we believe is of independent interest. We provide a robust analysis of the classic method of \emph{shift-and-invert} preconditioning to reduce eigenvector computation to \emph{approximately} solving a sequence of linear systems. We then apply variants of stochastic variance reduced gradient descent (SVRG) and additional recent advances in solving linear systems to achieve our claims. We believe our results suggest the generality and effectiveness of shift-and-invert based approaches and imply that further computational improvements may be reaped in practice.
\end{abstract}

%% file: introduction.tex

\section{Introduction}

Given a matrix $\bv{A} \in \mathbb{R}^{n \times d}$, computing the top eigenvector
of $\bv{A}^{\top}\bv{A}$ is a fundamental problem in computer science with
applications ranging from principal component analysis \cite{jolliffe2002principal}, to spectral clustering and learning of mixture models \cite{ng2002spectral,vempala2004spectral}, to pagerank computation \cite{page1999pagerank}, and a number of other graph related computations \cite{koren2003spectral,spielman2007spectral}.

In this paper we provide improved algorithms for computing the top eigenvector,
both in the \emph{offline} case, where the matrix $\bv{A}$ is given
explicitly as well as in the \emph{online} or \emph{statistical} case where we are simply given samples from a distribution $\mathcal{D}$ over
vectors $a \in \mathbb{R}^{d}$ and wish to compute the top eigenvector of $\E_{a\sim\dist} \left [aa^\top \right ]$, the covariance matrix of the distribution. 

Our algorithms are based on the classic idea of \emph{shift-and-invert} preconditioning for eigenvalue computation \cite{saad1992numerical}. We give a new robust analysis of the shifted-and-inverted power method, which allows us to efficiently reduce maximum eigenvector computation to \emph{approximately} solving a sequence of linear systems in the matrix $\lambda \bv I - \bv{A}^\top \bv {A}$ for some shift parameter $\lambda \approx \lambda_1(\bv{A})$. We then show how to solve these systems efficiently using variants of Stochastic Variance Reduced Gradient (SVRG) \cite{johnson2013accelerating} that optimize a convex function that is given as a sum of non-convex components.

\subsection{Our Approach}

The well known power method for computing the top eigenvector of a $\bv{A}^\top \bv{A}$ starts with a initial vector $x$ (often random) and repeatedly multiplies by $\bv{A}^\top \bv{A}$, eventually causing $x$ to converge to the top eigenvector. Assuming a random start vector, convergence requires $O \left ( \frac{\log(d/\epsilon)}{\gap} \right )$ iterations, where $\gap = \frac{\lambda_1-\lambda_2}{\lambda_1}$. 
The dependence on this gap is inherent to the power method and ensures that the largest eigenvalue is significantly amplified in comparison to the remaining values.

If the gap is small, one way to attempt to deal with this dependence is to replace $\bv{A}$ with a preconditioned matrix -- i.e. a matrix with the same top eigenvector but a much larger eigenvalue gap.
Specifically, let $\bv{B} = \lambda \bv{ I} - \bv{A}$ for some shift parameter $\lambda$. We can see that the smallest eigenvector of $\bv{B}$ (the largest eigenvector of $\bv{B}^{-1}$) is equal to the largest eigenvector of $\bv{A}$. Additionally, if $\lambda$ is near the largest eigenvalue of $\bv{A}$, $\lambda_1$, there will be a constant gap between the largest and second largest values of $\bv{B}^{-1}$. For example, if $\lambda = (1+\gap)\lambda_1$, then we will have $\lambda_1\left(\bv{B}^{-1} \right ) = \frac{1}{\lambda-\lambda_1} = \frac{1}{\gap \cdot \lambda_1}$ and $\lambda_2\left(\bv{B}^{-1} \right ) = \frac{1}{\lambda - \lambda_2} = \frac{1}{2\gap \cdot \lambda_1}$.

This constant factor gap means that running the power method on $\bv{B}^{-1}$ converges quickly to the top eigenvector of $\bv{A}$, specifically in $O \left (\log (d/\epsilon) \right )$ iterations. Of course, there is a catch -- each iteration of the shifted-and-inverted power method requires solving a linear system in $\bv{B}$. Furthermore, the condition number of $\bv{B}$ is proportional $\frac{1}{\gap}$, so as $\gap$ gets smaller, solving this linear system becomes more difficult for standard iterative methods.

Fortunately, the problem of solving linear systems is incredibly well studied and there are many efficient iterative algorithms we can adapt to apply $\bv{B}^{-1}$ approximately. In particular, we show how to accelerate the iterations of the shifted-and-inverted power method using variants of Stochastic Variance Reduced Gradient (SVRG) \cite{johnson2013accelerating}.
Due to the condition number of $\bv{B}$, we will not entirely avoid a $\frac{1}{\gap}$ dependence, however, we can separate this dependence from the input size $\nnz(\bv A)$.

Typically, stochastic gradient methods  are used to optimize convex functions that are given as the sum of many convex components. To solve a linear system $(\bv{M}^\top\bv{M}) x = b$ we minimize the convex function $f(x) = \frac{1}{2} x^\top (\bv{M}^\top\bv{M}) x + b^\top x$ with components $\psi_i(x) = \frac{1}{2} x^\top \left (m_i m_i^\top \right )x + \frac{1}{n}b^\top x$ where $m_i$ is the $i^{th}$ row of $\bv M$. Such an approach can be used to solve systems in $\bv{A}^\top \bv{A}$, however solving systems in $\bv{B} = \lambda \bv I - \bv A^\top \bv A$ requires more care. The components of our function are not so simple, and we require an analysis of SVRG that guarantees convergence even when some of these components are  \emph{non-convex}. We give a simple analysis for this setting, generalizing recent work appearing in the literature \cite{shalev2015sdca,csiba2015primal}.

Given fast approximate solvers for $\bv{B}$, the second main component required by our algorithmic framework is a new error bound for the shifted-and-inverted power method, showing that it is robust to the approximate linear system solvers, such as SVRG. 
We give a general robustness analysis, showing   
exactly what accuracy each system must be solved to, allowing for faster implementations using linear solvers with weaker guarantees. Our proofs center around the potential function $$G(x) \defeq
  \frac{\norm{\bv{P}_{v_1^{\perp}}x}_{\mb}}
  {\norm{\bv{P}_{v_1}x}_{\mb}}$$
where $\bv{P}_{v_1}$ and $\bv{P}_{v_1^{\perp}}$ are the projections onto the top eigenvector and its complement respectively. This function resembles tangent based potential functions used in previous work \cite{hardt2014noisy} except that the norms of the projections are measured over $\bv{B}$. For the exact power method, this is irrelevant -- it is not hard to see that progress is identical in both the $\ell_2$ and $\mb$ norms (see Lemma \ref{same_progress} of the Appendix). However, since $\norm{\cdot}_\mb$ is a natural norm for measuring the progress of linear system solvers for $\mb$, our potential function makes it possible to show that progress is also made when we compute $\bv{B}^{-1}x$ approximately up to some error $\xi$ with bounded $\norm{\xi}_\mb$.

\subsection{Our Results}

Our algorithmic framework described above offers several advantageous. Theoretically, we obtain improved running times for computing the top eigenvector. In the offline case, in Theorem \ref{main_offline_theorem} we give an algorithm running in time $O\left(\left [\nnz(\bv{A}) + \frac{d \nrank \bv{A}}{\gap^2}\right ] \cdot \left [ \log \frac{1}{\epsilon} + \log^2 \frac{d}{\gap} \right ]\right)$, where $\gap$ is the multiplicative gap between the largest and second largest eigenvalues, $\nrank(\bv{A}) = \norm{\bv A}_F^2 / \norm{\bv A}_2^2 \le \rank(\bv{A})$ is the stable rank of $\bv A$, and $\nnz(\bv{A})$ is the number of non-zero entries in the matrix. Up to log factors, our runtime is in many settings proportional to the input size $\nnz(\bv{A})$, and so is very efficient for large data matrices. In the case when $\nnz(\bv{A}) \le \frac{d\nrank(\bv{A})}{\gap^2}$ we also use the results of \cite{frostig2015regularizing, lin2015catalyst} to provide an accelerated runtime of $O\left(\left [\frac{\nnz(\bv{A})^{3/4}(d \nrank (\bv{A}))^{1/4}}{\sqrt{\gap}}\right ] \cdot \left [\log \frac{d}{\gap} \log \frac{1}{\epsilon} + \log^3 \frac{d}{\gap} \right ]\right)$, shown in Theorem \ref{accelerated_offline_theorem}.

Our algorithms return an approximate top eigenvector $x$ with $x^\top\bv{A}^\top \bv{A} x \ge (1-\epsilon)\lambda_1$. Note that, by choosing error $\epsilon \cdot \gap$, we can ensure that $x$ is actually close to $v_1$ -- i.e. that $|x^\top v_1| \ge 1-\epsilon$. Further, we obtain the same asymptotic runtime since $O\left (\log \frac{1}{\epsilon \gap} + \log^2 \frac{d}{gap} \right ) = O\left (\log \frac{1}{\epsilon} + \log^2 \frac{d}{gap} \right )$. We compare our runtimes with previous work in Table \ref{offline_table}. 

In the online case, in Theorem \ref{warmstart_online_theorem}, we show how to improve an $O(\gap)$ factor approximation to the top eigenvector to an $\epsilon$ approximation using $O\left (\frac{\nvar(\dist) \log\log(1/\epsilon)}{\gap^2} + \frac{\nvar(\dist)}{\gap \cdot \epsilon}\right)$ samples where $\nvar(\dist)$ is a natural upper bound on the variance of our distribution. Our algorithm is based off the streaming SVRG algorithm of \cite{frostig2014competing, lin2015catalyst}. It requires just $O(d)$ amortized time to process each sample, uses just $O(d)$ space, and is easy to parallelize. We can apply our result in a variety of regimes, using algorithms from the literature to obtain the initial $O(\gap)$ factor approximation and our algorithm to refine this solution. As shown in Table \ref{online_table}, this gives improved runtime and sample complexity over existing work. Notably, we give improved asymptotic sample complexity over known matrix concentration results for general distributions, and give the first streaming algorithm that is asymptotically optimal in the popular Gaussian spike model.

Our bounds hold for \emph{any} $\bv{A}$ or distribution $\mathcal{D}$.  In the offline case we require no initial knowledge of $\lambda_1(\bv{A})$, the eigenvalue gap, or the top eigenvector. We are hopeful that our online algorithm can also be made to work without such estimates.

Outside of our runtime results, our robust analysis of the shifted-and-inverted power method  provides new understanding of this well studied and widely implemented technique. It gives a means of obtaining provably accurate results when each iteration is implemented using fast approximate linear system solvers with rather weak accuracy guarantees.

In practice, this reduction between approximate linear system solving and eigenvector computation shows that regression libraries can be directly utilized to obtain faster running times for eigenvector computation in many cases. Furthermore, in theory we believe that our reduction suggests computational limits inherent in eigenvector computation as seen by the often easier-to-analyze problem of linear system solving. Indeed, in Section \ref{sec:lower}, we provide evidence that in certain regimes our statistical results are optimal.

We remark that during the preparation of our manuscript we found that
previously and independently Dan Garber and Elad Hazan had discovered
a similar technique using shift-and-invert preconditioning and SVRG for sums
of non-convex functions to improve the running time for offline eigenvector computation \cite{garber2015fast}.

\subsection{Previous Work}


\subsubsection*{Offline Eigenvector Computation}\label{previous_work_offline}

Due to its universal applicability, eigenvector computation in the offline case is extremely well studied. Classical methods, such as the QR algorithm, take roughly $O(nd^2)$ time to compute a full eigendecomposition. This can be accelerated to $O(nd^{\omega - 1})$, where $\omega < 2.373$ is the matrix multiplication constant \cite{williams2012multiplying,le2014powers}, however this is still prohibitively expensive for large matrices. Hence, faster iterative methods are often employed, especially when only the top eigenvector (or a few of the top eigenvectors) is desired.

As discussed, the popular power method requires $O \left ( \frac{\log(d/\epsilon)}{\gap} \right )$ iterations to converge to an $\epsilon$ approximate top eigenvector. Using Chebyshev iteration, or more commonly, the Lanczos method, this bound can be improved to $O \left ( \frac{\log(d/\epsilon)}{\sqrt{\gap}} \right )$ \cite{saad1992numerical}, giving total runtime of $O \left (\nnz(\bv A) \cdot \frac{\log(d/\epsilon)}{\sqrt{\gap}} \right )$.

Unfortunately, if $\bv{A}$ is very large and $\gap$ is small, this can still be quite expensive, and there is a natural desire to separate the $\frac{1}{\sqrt{\gap}}$ dependence from the $\nnz(\bv A)$ term. One approach is to use random subspace embedding matrices \cite{ailon2009fast,clarkson2013low} or fast row sampling algorithms \cite{cohen2015uniform}, which can be applied in $ O(\nnz(\bv A))$ time and yield a matrix $\bv{\tilde A}$ which is a good spectral approximation to the original. The number of rows in $\bv{\tilde A}$ depends only on the stable rank of $\bv A$ and the error of the embedding -- hence it can be significantly smaller than $n$. Applying such a subspace embedding and then computing the top eigenvector of $\bv{\tilde A}$ will require runtime $O \left (\nnz(\bv A) + \poly(\nrank(\bv{A}),\epsilon,\gap) \right )$, achieving the goal of reducing runtime dependence on the input size $\nnz(\bv A)$. Unfortunately, the dependence on $\epsilon$ will be significantly suboptimal -- such an approach cannot be used to obtain a linearly convergent algorithm. Further, the technique does not extend to online setting, unless we are willing to store a full subspace embedding of our sampled rows. 

Another approach, which we follow more closely, is to apply stochastic optimization techniques, which iteratively update an estimate to the top eigenvector, considering a random row of $\bv{A}$ with each update step. Such algorithms naturally extend to the online setting and have led to improved dependence on the input size for a variety of problems \cite{bottou2010large}. Using variance-reduced stochastic gradient techniques, \cite{shamir2015stochastic} achieves runtime $ O\left (\left (\nnz(\bv A) + \frac{dr^2n^2}{\gap^2\lambda_1^2} \right ) \cdot \log(1/\epsilon) \log\log(1/\epsilon) \right )$ for approximately computing the top eigenvector of a matrix with constant probability.  Here $r$ is an upper bound on the squared row norms of $\bv{A}$. In the \emph{best case}, when row norms are uniform, this runtime can be simplified to $O\left (\left (\nnz(\bv A) + \frac{d\nrank(\bv A)^2}{\gap^2} \right ) \cdot \log(1/\epsilon) \log\log(1/\epsilon) \right )$.

The result in \cite{shamir2015stochastic} makes an important contribution in separating input size and gap dependencies using stochastic optimization techniques.
Unfortunately, the algorithm requires an approximation to the eigenvalue gap and a starting vector that has a constant dot product with the top eigenvector. In \cite{shamir2015fast} the analysis is extended to a random initialization, however loses polynomial factors in $d$. Furthermore, the dependences on the stable rank and $\epsilon$ are suboptimal -- we improve them to $\nrank(\bv A)$ and $\log(1/\epsilon)$ respectively, obtaining true linear convergence.

\begin{table}[H]
\def\arraystretch{1.25}
\begin{center}
\begin{tabular}{|>{\centering}m{6.5cm}|c|}
\hline
\textbf{Algorithm}  & \textbf{Runtime}\\
\hline
Power Method & $O\left (\nnz(\bv{A})\frac{\log(d/\epsilon)}{\gap} \right )$ \\
\hline
Lanczos Method & $O\left (\nnz(\bv{A})\frac{\log(d/\epsilon)}{\sqrt{\gap}} \right )$ \\
\hline
Fast Subspace Embeddings \cite{clarkson2013low} Plus Lanczos  & $O\left (\nnz(\bv{A}) + \frac{d\nrank(\bv{A})}{\max \left \{ \gap^{2.5} \epsilon, \epsilon^{2.5} \right \} } \right )$ \\
\hline
SVRG \cite{shamir2015stochastic} (assuming uniform row norms and warm-start) & $O\left (\left (\nnz(\bv A) + \frac{d\nrank(\bv A)^2}{\gap^2} \right ) \cdot \log(1/\epsilon) \log\log(1/\epsilon) \right )$\\
\hline
\textbf{Theorem \ref{main_offline_theorem}} & $O\left(\left [\nnz(\bv{A}) + \frac{d \nrank \bv{A}}{\gap^2}\right ] \cdot \left [\log \frac{1}{\epsilon} + \log^2 \frac{d}{\gap} \right ]\right)$\\
\hline
\textbf{Theorem \ref{accelerated_offline_theorem}} & $O\left(\left [\frac{\nnz(\bv{A})^{3/4}(d \nrank(\bv{A}))^{1/4}}{\sqrt{\gap}}\right ] \cdot   \left[\log \frac{d}{\gap} \log \frac{1}{\epsilon} + \log^3 \frac{d}{\gap} \right ]\right)$\\
\hline
\end{tabular}
\caption{Comparision to Previous work on Offline Eigenvector Estimation. We give runtimes for computing $x$ such that $x^\top \bv{A}^\top\bv{A} x \ge (1-\epsilon)\lambda_1$.}
\label{offline_table}
\end{center}
\vspace{-1em}
\end{table}

\subsubsection*{Online Eigenvector Computation}

While in the offline case the primary concern is computation time, in the online, or statistical setting, research also focuses on minimizing the number of samples that we must draw from $\dist$ in order to achieve a given accuracy on our eigenvector estimate. Especially sought after are results that achieve asymptotically optimal accuracy as the sample size grows large.

While the result we give in Theorem \ref{warmstart_online_theorem} will work for any distribution parameterized by a variance bound, in this section, in order to more easily compare to previous work, we normalize $\lambda_1 = 1$ and assume we have the variance bound $\E_{a\sim\dist} \norm{(aa^\top)^2}_2 = O(d)$ along with the row norm bound $\norm{a}_2^2 \le O(d)$. Additionally, we compare runtimes for computing some $x$ such that $|x^\top v_1| \ge 1-\epsilon$, as this is the most popular guarantee studied in the literature. Theorem \ref{warmstart_online_theorem} is easily extended to this setting as obtaining $x$ with $x^T \bv{AA}^\top x \ge (1-\epsilon\cdot \gap) \lambda_1$ ensures $|x^\top v_1| \ge 1-\epsilon$. Our our algorithm requires $O\left ( \frac{d}{\gap^2 \epsilon} \right )$ samples to find such a vector under the assumptions given above.

 The simplest algorithm in this setting is to take $n$ samples from $\dist$ and compute the leading eigenvector of the empirical estimate $\widehat\E[a a^\top] = \frac{1}{n} \sum_{i=1}^n a_i a_i^\top$. By a matrix Bernstein bound, such as inequality Theorem 6.6.1 of \cite{tropp2015introduction}, $O\left ( \frac{d\log d}{\gap^2 \epsilon} \right )$ samples is enough to insure $\norm{\widehat\E[a a^\top] - \E[a a^\top]}_2 \le \sqrt{\epsilon}\gap$. 
By Lemma \ref{spectral_error_conversion} in the Appendix, this gives that, if $x$ is set to the top eigenvector of $\widehat\E[a a^\top] $ it will satisfy $|x^\top v_1| \ge 1-\epsilon$. Such an $x$ can then be approximated by applying any offline eigenvector algorithm to the empirical estimate.

A large body of work focuses on improving the computational and sample cost of this simple algorithm, under a variety of assumptions on $\dist$. The most common focus is on obtaining \emph{streaming algorithms}, in which the storage space is just $O(d)$ - proportional to the size of a single sample.

In Table \ref{online_table} we give a sampling of results in this area. All of these results rely on distributional assumptions at least as strong as those given above. In each setting, we can use the cited algorithm to first compute an $O(\gap)$ approximate eigenvector, and then refine this approximation to an $\epsilon$ approximation using $O\left ( \frac{d}{\gap^2 \epsilon} \right )$ samples by applying our streaming SVRG based algorithm. 
This allows us to obtain improved runtimes and sample complexities. Notably, by the lower bound shown in Section \ref{sec:lower}, in all settings considered in Table \ref{online_table}, we achieve optimal asymptotic sample complexity - as our sample size grows large, our $\epsilon$ decreases at an optimal rate. 
To save space, we do not include our improved runtime bounds in Table \ref{online_table}, however they are easy to derive by adding the runtime required by the given algorithm to achieve $O(\gap)$ accuracy, to $O\left (\frac{d^2}{\gap^{2} \epsilon}\right)$ -- the runtime required by our streaming algorithm.

The bounds given for the simple matrix Bernstein based algorithm described above,  Krasulina/Oja's Algorithm \cite{balsubramani2013fast}, and SGD \cite{shamir2015convergence}  require no additional assumptions, aside from those given at the beginning of this section.
The streaming results cited for \cite{mitliagkas2013memory} and \cite{hardt2014noisy} assume $a$ is generated from a
Gaussian spike model, where $a_i = \sqrt{\lambda_1}\gamma_i{v_1} + Z_i$
and $\gamma_i \sim \mathcal{N}(0, 1), Z_i \sim \mathcal{N}(0, I_d)$. We note that under this model, the matrix Bernstein results improve by a $\log d$ factor and so match our results in achieving asymptotically optimal convergence rate. The results of \cite{mitliagkas2013memory} and \cite{hardt2014noisy} sacrifice this optimality in order to operate under the streaming model. Our work gives the best of both works -- a streaming algorithm giving asymptotically optimal results. 

The streaming Alecton algorithm \cite{sa2015global} assumes $\E\norm{aa^\top \bv W a a^\top} \le O(1)\text{tr}(\bv W)$ for any symmetric $\bv W$ that commutes with $\E aa^\top$. This is a strictly stronger assumption than the assumption above that 
$\E_{a\sim\dist} \norm{(aa^\top)^2}_2 = O(d)$.



\begin{table}
\def\arraystretch{1.25}
\small
\begin{center}
\begin{tabular}{|c|c|c|c|c|}
\hline
\textbf{Algorithm} & \specialcell{\textbf{Sample}\\ \textbf{Size}} & \textbf{Runtime} & \textbf{Streaming?} & \specialcell{\textbf{Our Sample}\\ \textbf{Complexity}}\\ 
\hline
\specialcell{Matrix Bernstein plus \\ Lanczos (explicitly forming\\ sampled matrix)} & $O\left (\frac{d\log d}{gap^2 \epsilon}\right) $ & $O\left(\frac{d^3\log d}{gap^2\epsilon}\right)$ & $\times$ & $O\left (\frac{d\log d}{gap^3} + \frac{d}{gap^2 \epsilon}\right) $ \\ 
\hline
\specialcell{Matrix Bernstein plus \\Lanczos (iteratively applying\\ sampled matrix)} & $O\left (\frac{d\log d}{gap^2 \epsilon}\right) $ & $O\left (\frac{d^2\log d\cdot \log(d/\epsilon)}{gap^{2.5}\epsilon}\right)$ & $\times$ & $O\left (\frac{d\log d}{gap^3} + \frac{d}{gap^2 \epsilon}\right) $\\ 
\hline
\specialcell{Memory-efficient PCA \\\cite{mitliagkas2013memory, hardt2014noisy}} & $O\left(\frac{d\log (d/\epsilon)}{gap^3 \epsilon}\right )$ & $O\left (\frac{d^2\log (d/\epsilon)}{gap^3 \epsilon}\right)$ &  $\surd$ & $O\left(\frac{d\log (d/\gap)}{gap^4} + \frac{d}{gap^2 \epsilon}\right )$\\ 
\hline
Alecton \cite{sa2015global} & $O(\frac{d\log (d/\epsilon)}{gap^2 \epsilon})$ & $O(\frac{d^2\log (d/\epsilon)}{gap^2 \epsilon})$ & $\surd$ & $O(\frac{d\log (d/\gap)}{gap^3} + \frac{d}{gap^2 \epsilon})$ \\ 
\hline
\specialcell{Krasulina / Oja's \\Algorithm \cite{balsubramani2013fast}} & $O(\frac{d^{c_1}}{gap^2 \epsilon^{c_2}})$ &$O(\frac{d^{c_1+1}}{gap^2 \epsilon^{c_2}})$  & $\surd$ & $O(\frac{d^{c_1}}{gap^{2+c_2}} + \frac{d}{gap^2 \epsilon})$ \\ 
\hline
SGD \cite{shamir2015convergence} & $O(\frac{d^3\log (d/\epsilon)}{\epsilon^2})$ & $O(\frac{d^4\log (d/\epsilon)}{\epsilon^2})$ & $\surd$ & $O\left (\frac{d^3\log (d/\gap)}{\gap^2} +\frac{d}{\gap^2 \epsilon} \right ) $ \\
\hline
\end{tabular}
\caption{Summary of existing work on Online Eigenvector Estimation and improvements given by our results. Runtimes are for computing $x$ such that $|x^\top v_1| \ge 1-\epsilon$. For each of these results we can obtain improved running times and sample complexities by running the algorithm to first compute an $O(\gap)$ approximate eigenvector, and then running our algorithm to obtain an $\epsilon$ approximation using an additional $O\left ( \frac{d}{\gap^2 \epsilon} \right )$ samples, $O(d)$ space, and $O(d)$ work per sample.
	}
\label{online_table}
\end{center}
\end{table}

\subsection{Paper Organization}
\begin{description}
\item[Section \ref{prelims}] Review problem definitions and parameters for our runtime and sample bounds.
\item[Section \ref{framework}] Describe the shifted-and-inverted power method and show how it can be implemented using approximate system solvers.
\item[Section \ref{sec:offline}] Show how to apply SVRG to solve systems in our shifted matrix, giving our main runtime results for eigenvector computation in the offline setting.
\item[Section \ref{sec:online}] Show how to use an online variant of SVRG to run the shifted-and-inverted power method, giving our main sampling complexity and runtime results in the statistical setting.
\item[Section \ref{parameter_free}] Show how to efficiently estimate the shift parameters required by our algorithms, completing their analysis.
\item[Section \ref{sec:lower}] Give a lower bound in the statistical setting, showing that our results are asymptotically optimal.
\end{description}

%% file: preliminaries.tex
\section{Preliminaries}\label{prelims}

We bold all matrix variables. 
We use $[n] \defeq \{1,...,n\}$. For a symmetric positive semidefinite (PSD) matrix $\bv{M}$ we let $\norm x_{\bv{M}} \defeq \sqrt{x^{\top}\bv{M} x}$ and we let $\lambda_1(M), ..., \lambda_n(M)$ denote its eigenvalues in decreasing order.
We use $\bv{M} \preceq \bv{N}$ to denote the condition that $x^{\top}\bv{M} x\leq x^{\top}\bv{N} x$
for all $x$. 

\subsection{The Offline Problem }

We are given a matrix $\ma\in\R^{n\times d}$ with rows $a^{(1)},...,a^{(n)}\in\R^{d}$
and wish to compute an approximation the top eigenvector of $\mSigma\defeq\ma^{\top}\ma$. Specifically for some error parameter $\epsilon$ we want a unit vector $x$ such that $x^\top \bv{\Sigma} x \ge (1-\epsilon) \lambda_1(\bv{\Sigma})$.

\subsection{The Statistical Problem}

We are given $n$ independent samples from a distribution $\dist$
on $\R^{d}$ and wish to compute the top eigenvector of $\mSigma\defeq\E_{a\sim\dist} \left [aa^\top \right ]$. Again, for some error parameter $\epsilon$ we want to return a unit vector $x$ such that $x^\top \bv{\Sigma} x \ge (1-\epsilon) \lambda_1(\bv{\Sigma})$.

\subsection{Problem Parameters}

We parameterize the running time of our algorithm in terms of several natural properties of $\ma$, $\dist$, and $\mSigma$. We let $\lambda_{1},...,\lambda_{d}$ denote the eigenvalues of $\mSigma$
in decreasing order and we let $v_1,... ,v_d$ denote their corresponding
eigenvectors. We define the \emph{eigenvalue gap} by $\gap\defeq\frac{\lambda_{1}-\lambda_{2}}{\lambda_{1}}$.

We use the following additional parameters to provide running times for the offline and statistical problems respectively:

\begin{itemize}
	\item \textbf{Offline Problem}: We let  $\nrank(\ma) \defeq\sum_{i}\frac{\lambda_{i}}{\lambda_{1}} = \frac{\norm{\ma}_F^2}{\norm{\ma}_2^2}$ denote the stable rank of $\ma$. Note that we always have $\nrank(\ma) \le \rank(\ma)$. We let $\nnz(\ma)$ denote the number of non-zero entries in $\ma$.
	\item \textbf{Online Problem}: We let $\nvar(\dist) \defeq \frac{\norm{\E_{a \sim \dist}\left [ \left (a a^\top \right )^2 \right ]}_2}{\norm{\E_{a \sim \dist} (aa^\top)}_2^2} = \frac{\norm{\E_{a \sim \dist}\left [ \left (a a^\top \right )^2 \right ]}_2}{\lambda_1^2}$ denote a natural upper bound on the variance of $\dist$ in various settings. Note that $\nvar(\dist) \ge 1$.
\end{itemize}

%% file: framework.tex

\section{Framework}\label{framework}

Here we develop our robust shift-and-invert framework. In Section~\ref{sec:framework:basics} we provide a basic overview of the framework and in Section~\ref{sec:framework:potential} we introduce the potential function we use to measure progress of our algorithms. In Section~\ref{sec:framework:power-iteration} we show how to analyze the framework given access to an exact linear system solver and in Section~\ref{sec:framework:approximate-power} we strengthen this analysis to work with an inexact linear system solver. Finally, in Section~\ref{sec:framework:init} we discuss initializing the framework.

\subsection{Shifted-and-Inverted Power Method Basics}
\label{sec:framework:basics}

We let $\mb_{\lambda}\defeq\lambda\mI-\mSigma$ denote the shifted matrix that we will use in our implementation of the shifted-and-inverted power method. As discussed, in order for $\bv{B}_\lambda^{-1}$ to have a large eigenvalue gap, $\lambda$ should be set to $(1+c \cdot \gap) \lambda_1$ for some constant $c \geq 0$. Throughout this section we assume that we have a crude estimate of $\lambda_1$ and $\gap$ and fix $\lambda$ to be a value satisfying  $\left (1 + \frac{\gap}{150}\right)\lambda_1 \le \lambda \le \left (1 + \frac{\gap}{100}\right)\lambda_1$. (See Section~\ref{parameter_free} for how we can compute such a $\lambda$). For the remainder of this section we work with such a fixed value of $\lambda$ and therefore for convenience denote $\bv{B}_\lambda$ as $\bv{B}$.

Note that $\lamiBinv{i} = \frac{1}{\lambda_i(\mb)} = \frac{1}{\lambda-\lambda_i}$ and so $\frac{\lamiBinv{1}}{\lamiBinv{2}} = \frac{\lambda-\lambda_2}{\lambda-\lambda_1} \ge \frac{\gap}{\gap/100} = 100.$ This large gap will ensure that, assuming the ability to apply $\bv{B}^{-1}$, the power method will converge very quickly. In the remainder of this section we develop our error analysis for the shifted-and-inverted power method which demonstrates that approximate application of $\bv{B}^{-1}$ in each iteration in fact suffices.

\subsection{Potential Function}
\label{sec:framework:potential}

Our analysis of the power method focuses on the objective of maximizing
the Rayleigh quotient, $x^\top \mSigma x$ for unit vector $x$. Note that as the following lemma shows, this has a directly correspondence to the error in maximizing $|v_1^\top x|$:

\begin{lemma}[Bounding Eigenvector Error by Rayleigh Quotient]
\label{lem:ray_to_evec}
For a unit vector $x$ let $\epsilon = \lambda_1 - x^\top \bv{\Sigma} x$. If $\epsilon \leq \lambda_1 \cdot \gap$ then
\[
\left | v_1^\top x \right |
\geq \sqrt{1 - \frac{\epsilon}{\lambda_1 \cdot \gap}}.
\]
\end{lemma}

\begin{proof}
Among all unit vectors $x$ such that $\epsilon = \lambda_1 - x^\top \mSigma x$, a minimizer of $\left | v_1^\top x\right |$ has the form $x = (\sqrt{1 - \delta^2}) v_1 + \delta v_2$ for some $\delta$. We have
\begin{align*}
\epsilon
&=
\lambda_1 - x^\top \bv \Sigma x
= 
\lambda_1 - \lambda_1 (1 - \delta^2) - \lambda_2 \delta^2
=
(\lambda_1 - \lambda_2) \delta^2.
\end{align*}
Therefore by direct computation,
\begin{align*}
\left | v_1^\top x \right | = \sqrt{1 - \delta^2} = \sqrt{1 - \frac{\epsilon}{\lambda_1-\lambda_2}} = \sqrt{1 - \frac{\epsilon}{\lambda_1 \cdot \gap}} ~ .
\end{align*}
\end{proof}

In order to track the progress of our algorithm we use a more
complex potential function than just the Rayleigh quotient error, $\lambda_1 - x^\top \bv{\Sigma}
x$.  Our potential function $G$ is defined for $x \neq 0$ by
\begin{align*}
	G(x) \defeq
  \frac{\norm{\bv{P}_{v_1^{\perp}}\left(x\right)}_{\mb}}
  {\norm{\bv{P}_{v_1}(x)}_{\mb}} 
\end{align*}
where $\bv{P}_{v_1}(x)$ and $\bv{P}_{v_1^{\perp}}\left(x\right)$ denote the
projections of $x$ in the direction of $v_1$ and on the subspace
orthogonal to $v_1$ respectively. Equivalently, we have that:
\begin{align}\label{gequiv}
	G(x) =
  \frac{\sqrt{\norm{x}_{\mb}^{2}-\left(v_1^{\top}\mb^{1/2}x\right)^{2}}}{\abs{v_1^\top
  \mb^{1/2}x}} =  \frac{\sqrt{\sum_{i\geq 2} \frac{\alpha_i^2}{ \lambda_{i}(\bv{B}^{-1})}}}{\sqrt{\frac{\alpha_1^2}{ \lambda_{1}(\bv{B}^{-1})}}}.
\end{align}
where $\alpha_i = v_i^\top x$.

When the Rayleigh quotient error $\epsilon = \lambda_1 - x^\top\mSigma
x$ of $x$ is small, we can show a strong relation between $\epsilon$  and $G(x)$. We prove this in two parts. First we prove a technical lemma, Lemma~\ref{lem:potfunc1}, that we will use several time for bounding the numerator of $G$ and then we prove the connection in Lemma~\ref{lem:rayquot-potential}. 

\begin{lemma} \label{lem:potfunc1} For a unit vector $x$ and $\epsilon = \lambda_1 - x^\top \mSigma x$ if $\epsilon \leq \lambda_1 \cdot \gap$ then 
\[
\epsilon \leq 
x^\top \mb x - (v_1^\top \mb x) (v_1^\top x)
\leq \epsilon \left (1 + 
 \frac{\lambda - \lambda_1}{\lambda_1 \cdot \gap}\right).
\]
\end{lemma}

\begin{proof}
Since $\mb = \lambda \mI - \mSigma$ and since $v_1$ is an eigenvector of 
$\mSigma$ with eigenvalue $\lambda_1$ we have
\begin{align*}
x^\top \mb x - (v_1^\top \mb x) (v_1^\top x)
&=
\lambda \norm{x}_2^2 - x^\top \mSigma x
- (\lambda v_1^\top x - v_1^\top \mSigma x) (v_1^\top x)
\\
&=
\lambda - \lambda_1 + \epsilon -
(\lambda v_1^\top x - \lambda_1 v_1^\top x) (v_1^\top x)
\\
&=
(\lambda - \lambda_1) \left(1 - (v_1^\top x)^2\right) + \epsilon.
\end{align*}
Now by Lemma~\ref{lem:ray_to_evec} we know that $| v_1^\top x |
\geq \sqrt{1 - \frac{\epsilon }{\lambda_1 \cdot \gap}}$, giving us the upper bound. Furthermore, since trivially $\left |v_1^\top x\right | \leq 1$ and $\lambda - \lambda_1 > 0$, we have the lower bound.
\end{proof}

\begin{lemma}[Potential Function to Rayleigh Quotient Error Conversion]\label{lem:rayquot-potential}
	For a unit vector $x$ and $\epsilon = \lambda_1 - x^\top \mSigma x$ if $\epsilon \leq \frac{1}{2}\lambda_1 \cdot \gap$, we have:
	\begin{align*}
		\frac{\epsilon}{\lambda-\lambda_1} \leq G(x)^2 \leq \left(1+\frac{\lambda-\lambda_1}{\lambda_1 \cdot \gap}\right) \left(1+\frac{2\epsilon}{\lambda_1 \cdot \gap}\right)\frac{\epsilon}{\lambda-\lambda_1}.
	\end{align*}
\end{lemma}
\begin{proof}
	Since $v_1$ is an eigenvector of $\mb$, we can write $G(x)^2 = \frac{x^\top \mb x - (v_1^\top \mb x) (v_1^\top x)}{(v_1^\top \mb x) (v_1^\top x)}$. Lemmas~\ref{lem:ray_to_evec} and~\ref{lem:potfunc1} then give us:
	\begin{align*}
		\frac{\epsilon}{\lambda-\lambda_1} \leq G(x)^2 \leq \left(1+\frac{\lambda-\lambda_1}{\lambda_1 \cdot \gap}\right)\frac{\epsilon}{\left(\lambda-\lambda_1\right) \left(1-\frac{\epsilon}{\lambda_1 \cdot \gap}\right)}.
	\end{align*}
	Since $\epsilon \leq \frac{1}{2}\lambda_1 \cdot \gap$, we have $\frac{1}{1-\frac{\epsilon}{\lambda_1 \cdot \gap}} \leq 1+\frac{2\epsilon}{\lambda_1 \cdot \gap}$. This proves the lemma.
\end{proof}

\subsection{Power Iteration}
\label{sec:framework:power-iteration}

Here we show that the shifted-and-inverted power iteration in fact makes progress with respect to our objective function given an exact linear system solver for $\mb$. Formally, we show that applying $\mb^{-1}$ to a vector $x$ decreases the potential function $G(x)$ geometrically.

\begin{theorem}\label{thm:powermethod}
	Let $x$ be a unit vector with $\inprod{x,v_1} \neq 0$ and let $\xtilde = \mb^{-1}w$, i.e. the power method update of $\mb^{-1}$ on $x$. Then, we have:
	\begin{align*}
		G(\xtilde) \leq \frac{\lamiBinv{2}}{\lamiBinv{1}} G(x)
		\leq \frac{1}{100} G(x).
	\end{align*}
\end{theorem}
Note that $\xtilde$ may no longer be a unit vector. However, $G(\xtilde, v_1) = G(c \xtilde, v_1)$ for any scaling parameter $c$, so the theorem also holds for $\xtilde$ scaled to have unit norm.
\begin{proof}
	Writing $x$ in the eigenbasis, we have $x=\sum_i \alpha_i v_i$ and  $\xtilde = \sum_i \alpha_i \lamiBinv{i}v_i$. Since $\inprod{x,v_1} \neq 0$, $\alpha_1 \neq 0$  and by the equivalent formation of $G(x)$ given in \eqref{gequiv}:
	\begin{align*}
		G(\xtilde) = \frac{\sqrt{\sum_{i\geq 2} \alpha_i^2 \lambda_{i}(\bv{B}^{-1})}}{\sqrt{\alpha_1^2 \lambda_{1}(\bv{B}^{-1})}}
		\leq \frac{\lamiBinv{2}}{\lamiBinv{1}} \cdot \frac{\sqrt{\sum_{i\geq 2} \frac{\alpha_i^2}{\lambda_{i}(\bv{B}^{-1})}}}{\sqrt{\frac{\alpha_1^2}{\lambda_1(\bv{B}^{-1})}}}
		= \frac{\lamiBinv{2}}{\lamiBinv{1}} \cdot G(x) ~.
	\end{align*}
Recalling that $\frac{\lamiBinv{1}}{\lamiBinv{2}} = \frac{\lambda-\lambda_2}{\lambda-\lambda_1} \ge \frac{\gap}{\gap/100} = 100$ yields the result. 
\end{proof}
The challenge in using the above theorem, and any traditional analysis of the shifted-and-inverted power method, is that we don't actually have access to $\mb^{-1}$. In the next section we show that the shifted-and-inverted power method is robust -- we still make progress on our objective function even if we only approximate $\mb^{-1}x$ using a fast linear system solver.

\subsection{Approximate Power Iteration}
\label{sec:framework:approximate-power}

We are now ready to prove our main result on the shifted-and-inverted power method using approximate linear system solves at each iteration. In words, we show that each iteration of the power method makes constant factor expected progress on our potential function assuming we:
\begin{enumerate}
\item Start with a sufficiently good $x$ and an approximation of $\lambda_1$ 
\item Apply $\bv{B}^{-1}$ approximately using a system solver such that the function error (or distance to $\bv{B}^{-1} x$ in the $\bv{B}$ norm) is sufficiently small \emph{in expectation}.
\item Estimate Rayleigh quotients over $\bv{\Sigma}$ well enough to only accept updates that do not hurt progress on the objective function too much.
\end{enumerate}

This third assumption is necessary since the second assumption is quite weak. An expected progress bound on the linear system solver allows, for example, the solver to occasionally return a solution that is entirely orthogonal to $v_1$, causing us to make unbounded backwards progress on our potential function. The third assumption allows us to reject possibly harmful updates and ensure that we still make progress in expectation. In the offline setting, we can access $\bv{A}$ and are able able to compute Rayleigh quotients exactly in time $\nnz(\bv{A})$ time. However, we only assume the ability to estimate quotients as in the online setting we only have access to $\mSigma$ through samples from $\dist$.

While, our general theorem for the approximate power iteration, Theorem~\ref{thm:powermethod-perturb}, assumes that we can solve linear systems to some absolute accuracy in expectation, this is not the standard assumption for many linear system solvers. Many fast iterative linear system solvers assume an initial approximation to $\bv{B}^{-1} x$ and then show that the quality
of this approximation is then improved geometrically in each iteration of the algorithm. In Corollary~\ref{cor:constant_factor_corollary} we show how to find a  coarse initial approximation to $\bv{B}^{-1} x$, in fact just approximating $\bv{B}^{-1}$ with $\frac{1}{x^\top \mb x} x$. Moreover, using this course approximation in Corollary~\ref{cor:constant_factor_corollary} we show that Theorem~\ref{thm:powermethod-perturb} actually implies that it suffices to just make a fixed relative amount of progress in solving the linear system. 

Note that in both claims we measure error of the linear system solver using $\norm{\cdot}_\mb$. This is a natural norm in which geometric convergence is shown for many linear system solvers and directly corresponds to the function error of minimizing $f(w) = \frac{1}{2} w^\top \mb w - w^\top x$ to compute $\mb^{-1} x$.

\begin{theorem}[Approximate Shifted-and-Inverted Power Iteration -- Warm Start]\label{thm:powermethod-perturb}
	Let $x=\sum_i \alpha_i v_i$ be a unit vector such that $G(x) \leq \frac{1}{\sqrt{10}}$. Suppose we know some shift parameter $\lambda$ with $\left (1 + \frac{\gap}{150} \right ) \lambda_1< \lambda \le \left (1 + \frac{\gap}{100} \right ) \lambda_1$ and an estimate $\lambdah_1$ of $\lambda_1$ such that
	$
		\frac{10}{11} \left(\lambda-\lambda_1\right) \leq \lambda-\lambdah_1 \leq \lambda-\lambda_1
	$.
	Furthermore, suppose we have a subroutine $\mathrm{solve}(\cdot)$ that on any input $x$
	\begin{align*}
		\expec{\norm{\solve{x} - \mb^{-1}x}_{\mb}} \leq  \frac{c_1}{1000}
		\sqrt{\lambda_1(\bv{B}^{-1})},
	\end{align*}
	for some $c_1 < 1$,
	and a subroutine $\rayquoth{\cdot}$ that on any input $x \neq 0$
	\begin{align*}
		\abs{\rayquoth{x} - \rayquot(x)}
			\leq \frac{1}{30}\left(\lambda-\lambda_1\right) \;  \text{ for all nonzero } \; x \in \R^d.
	\end{align*}
	where $\rayquot(x) \defeq \frac{x^\top \mSigma x}{x^\top x}$.

	Then the following update procedure:
	\begin{align*}
		\text {Set }\xhat = \solve{x},
	\end{align*}
	\begin{align*}
		\text {Set } \xtilde = \left\{
		\begin{array}{cc}
		\xhat& \mbox{ if } \left\{
		\begin{array}{c}
			\rayquoth{\xhat} \geq \lambdah_1 - \left(\lambda-\lambdah_1\right)/6 \mbox{ and } \\
			\norm{\xhat}_2 \geq \frac{2}{3}\frac{1}{\lambda-\lambdah_1}
		\end{array} \right. \\
		x & \mbox{ otherwise,}
		\end{array}
		\right.
	\end{align*}
	satisfies the following:
	\begin{itemize}
		\item	$G(\xtilde)\leq \frac{1}{\sqrt{10}}$ and
		\item	$\expec{G(\xtilde)}\leq \frac{3}{25} G(x) + \frac{c_1}{500}$.
	\end{itemize}
\end{theorem}
That is, not only do we decrease our potential function by a constant factor in expectation, but we are guaranteed that the potential function will never increase beyond $1/\sqrt{10}$.
\begin{proof}
	The first claim follows directly from our choice of $\xtilde$ from $x$ and $\xhat$. 
	If $\xtilde = x$, it holds trivially by our assumption that $G(x) \leq \frac{1}{\sqrt{10}}$.  Otherwise, $\xtilde = \xhat$ and we know that 
	\begin{align*}
	\lambda_1 - \rayquot\left(\xhat\right) 
	&\leq \lambdah_1 - \rayquot\left(\xhat\right) \leq \lambdah_1 - \rayquoth{\xhat} + \abs{\rayquoth{\xhat} - \rayquot\left(\xhat\right)} 
	\\
	&\leq \frac{\lambda-\lambdah_1}{6} + \frac{\lambda-\lambda_1}{30} 
	\leq \frac{\lambda-\lambda_1}{5} \le \frac{\lambda_1 \cdot \gap}{500} ~.
	\end{align*} 
	The claim then follows from Lemma~\ref{lem:rayquot-potential} as
	\begin{align*}
		G(\xhat)^2 
		&\leq 
		\left(1 + \frac{\lambda - \lambda_1}{\lambda_1 \cdot \gap}\right) 
		\left(1 + \frac{2\left(\lambda_1 - \rayquot\left(\xhat\right)\right)}{\lambda_1 \cdot \gap}\right)
		\frac{\lambda_1 - \rayquot\left(\xhat\right)}{\lambda-\lambda_1} \\
		&\leq \frac{101}{100}\cdot \frac{251}{250} \cdot
		\frac{\left(\frac{\lambda_1 \cdot \gap}{500}\right)}
		{\left(\frac{\lambda_1 \cdot \gap}{150}\right)} 
		\le \frac{1}{\sqrt{10}} ~.
	\end{align*}
	
	All that remains is to show the second claim, that $\expec{G(\xtilde)}\leq \frac{3}{25} G(x) + \frac{4c_1}{1000}$. Let $\mcalF$ denote the event that we accept our iteration and set $x= \xhat = \solve{x}$. That is:
	\begin{align*}
	\mcalF &\defeq \set{\rayquoth{\xhat} \geq \lambdah_1 - 
		\frac{\lambda - \lambdah_1}{6}} \cup \set{\norm{\xhat}_2 \geq \frac{2}{3}\frac{1}{\lambda-\lambdah_1}}.
	\end{align*}
	Using our bounds on $\lambdah_1$ and $\rayquoth{\cdot}$, we know that $\rayquoth{x} \leq \rayquot(x) + (\lambda - \lambda_1) / 30$ and $\lambda - \lambdah_1 \leq  \lambda - \lambda_1$.
	Therefore, since $-1/6 - 1/30 \geq -1/2$ we have
	\begin{align*}
	\mathcal{F} &\subseteq \set{\rayquot\left(\xhat\right) \geq \lambda_1 - \left(\lambda-\lambda_1\right)/2} \cup \set{\norm{\xhat}_2 \geq \frac{2}{3}\frac{1}{\lambda-\lambda_1}},
	\end{align*}
	
	We will complete the proof in two steps. First we let $\xi \defeq \xhat - \mb^{-1}x$ and show that assuming $\mcalF$ is true then $G(\xhat)$ and
    $\norm{\xi}_{\mb}$ are linearly related, i.e. expected error bounds on $\norm{\xi}_{\mb}$ correspond to expected error bounds on $G(\xhat)$. Second, we bound the probability that $\mcalF$ does not occur and bound error incurred in this case. Combining these yields the result. 
    
    To show the linear relationship in the case where $\mcalF$ is true, first note Lemma~\ref{lem:ray_to_evec} shows that in this case $\abs{v_1^\top \frac{\xhat}{\norm{\xhat}_2}} \geq \sqrt{1-\frac{\lambda_1-\rayquot(\xhat)}{\lambda_1 \cdot \gap}} \geq \frac{3}{4}$. Consequently,
    \[
    \norm{\bv{P}_{v_1}\left(\xhat \right)}_{\mb}
    = \abs{v_1^\top\xhat}\sqrt{\lambda-\lambda_1}
    = \abs{v_1^\top \frac{\xhat}{\norm{\xhat}_2}}\cdot \norm{\xhat}\sqrt{\lambda-\lambda_1} 
     \geq \frac{3}{4}\cdot \frac{2}{3} \frac{1}{\sqrt{\lambda-\lambda_1}}
     =
     \frac{\sqrt{\lambda_1(\bv{B}^{-1})}}{2} ~.
    \]
    However, 
    \[
    \norm{\bv{P}_{v_1^{\perp}}\left(\xhat \right)}_{\mb}
	\leq \norm{\bv{P}_{v_1^{\perp}}\left(\mb^{-1} x \right)}_{\mb} + \norm{\bv{P}_{v_1^{\perp}}\left(\xi\right)}_\mb
		\leq \norm{\bv{P}_{v_1^{\perp}}\left(\mb^{-1} x \right)}_{\mb} + \norm{\xi}_{\mb} 
	\]
	and by Theorem~\ref{thm:powermethod} and the definition of $G$ we have
	\[
	\norm{\bv{P}_{v_1^{\perp}}\left(\mb^{-1} x \right)}_{\mb}
	= \norm{\bv{P}_{v_1}\left(\mb^{-1} x \right)}_{\mb} \cdot G(\mb^{-1} x)
	\leq \left(\abs{\inprod{x,v_1}} \sqrt{\lambda_1(\bv{B}^{-1})} \right)
	\cdot \frac{G(x)}{100} ~.
	\]
	Taking expectations, using that $\abs{\inprod{x,v_1}} \leq 1$, and combining these three inequalities yields
	\begin{equation}
	\expec{G(\xhat) \middle\vert  \mcalF}
	=
	\expec{\frac{\norm{\bv{P}_{v_1^{\perp}}\left(\mb^{-1} x \right)}_{\mb}}{\norm{\bv{P}_{v_1}\left(\mb^{-1} x \right)}_{\mb}} \middle\vert  \mcalF}
	\leq \frac{G(x)}{50}  
	+ 2 \frac{ \expec{\norm{\xi}_{\mb} 
			 \middle\vert \mcalF}}{\sqrt{\lambda_1(\bv{B}^{-1})}}
	\label{eqn:expec-potential}
	\end{equation}
	So, conditioning on making an update and changing $x$ (i.e. $\mathcal{F}$ occurring), we see that our potential function changes exactly as in the exact case (Theorem \ref{thm:powermethod}) with additional additive error due to our inexact linear system solve.

	Next we upper bound $\prob{\mcalF}$ and use it to compute $\expec{\norm{\xi}_{\mb}\middle\vert \mcalF}$. We will show that
	 $$\mcalG \defeq \set{\norm{\xi}_{\mb} \leq \frac{1}{100}\cdot \sqrt{\lamiBinv{1}} } \subseteq \mcalF$$ which then implies by Markov inequality that
	\begin{align}
		\prob{\mcalF} &\geq \prob{\norm{\xi}_{\mb} \leq \frac{1}{100}\cdot \sqrt{\lamiBinv{1}} } 
		\geq 1- \frac{\expec{\norm{\xi}_{\mb}}}{\frac{1}{100}\cdot \sqrt{\lamiBinv{1}} } \geq \frac{9}{10},\label{eqn:PF}
	\end{align}
	where we used the fact that $\E [\norm{\xi}_{\mb}] \le \frac{c_1}{1000} \sqrt{\lambda_1(\bv{B}^{-1})}$ for some $c_1 < 1$.
	
	Let us now show that $\mcalG \subseteq \mcalF$. Suppose $\mcalG$ is occurs. We can bound $\norm{\xhat}_2$ as follows:
	\begin{align}
		\norm{\xhat}_2 
		& \geq \norm{\mb^{-1}x}_2 - \norm{\xi}_2
		\geq \norm{\mb^{-1}x}-\sqrt{\lamiBinv{1}}\norm{\xi}_{\mb} \nonumber \\
		&\geq \abs{\alpha_1}\lamiBinv{1} - \frac{1}{100}\cdot \lamiBinv{1} \nonumber \\
		&= \frac{1}{\lambda-\lambda_1} \left(\abs{\alpha_1} - \frac{1}{100}\right)
		\geq \frac{3}{4}\frac{1}{\lambda-\lambda_1},\label{eqn:f2}
	\end{align}
	where we use Lemmas~\ref{lem:potfunc1} and~\ref{lem:rayquot-potential} to conclude that $\abs{\alpha_1} \geq \sqrt{1-\frac{1}{10}}$. We now turn to showing the Rayleigh quotient condition required by $\mcalF$. In order to do this, we first bound $\xhat^\top \mb \xhat - \left(v_1^\top \mb \xhat\right)\left(v_1^\top \xhat\right)$ and then use Lemma~\ref{lem:potfunc1}. We have:
	\begin{align*}
	\sqrt{\xhat^\top \mb \xhat - \left(v_1^\top \mb \xhat\right)\left(v_1^\top \xhat\right)} 
	&= \norm{\bv{P}_{v_1^{\perp}}\left(\xhat \right)}_\mb 
	\leq \norm{\bv{P}_{v_1^{\perp}}\left(\mb^{-1}x\right)}_\mb + \norm{\bv{P}_{v_1^{\perp}}\left( \xi \right)}_\mb \\
	&\leq \sqrt{\sum_{i\geq 2} \alpha_i^2 \lamiBinv{i}} + \frac{1}{100}\cdot \sqrt{\lamiBinv{1}} \\
	&\leq \sqrt{\lamiBinv{2}} + \frac{1}{100}\cdot \sqrt{\lamiBinv{1}}
	\leq \frac{1}{9}\sqrt{\lambda-\lambda_1},
	\end{align*}
	where we used the fact that $\lamiBinv{2}\leq \frac{1}{100}\lamiBinv{1}$ since $\lambda \le \lambda_1 + \frac{\gap}{100}$ in the last step. Now, using Lemma~\ref{lem:potfunc1} and the bound on $\norm{\xhat}_2$, we conclude that
	\begin{align}
		\lambdah_1 - \rayquoth{\xhat}
		&\leq \lambda_1 - \rayquot\left(\xhat\right) + \abs{\rayquot\left(\xhat\right)-\rayquoth{\xhat}} + \lambdah_1 - \lambda_1\nonumber \\
		&\leq \frac{\xhat^\top \mb \xhat - \left(v_1^\top \mb \xhat\right)\left(v_1^\top \xhat\right)}{\norm{\xhat}^2_2} + \frac{\lambda-\lambda_1}{30} + \frac{\lambda-\lambda_1}{11} \nonumber \\
		&\leq \frac{1}{81\left(\lambda-\lambda_1\right)} \cdot \frac{16}{9}\left(\lambda-\lambda_1\right)^2 + \frac{\lambda-\lambda_1}{8} \nonumber \\
		&\leq \left(\lambda-\lambda_1\right)/6
		\leq \left(\lambda-\lambdah_1\right)/4. \label{eqn:f1}
	\end{align}
	Combining~\eqref{eqn:f2} and~\eqref{eqn:f1} shows that $\mcalG \subseteq \mcalF$ there by proving~\eqref{eqn:PF}.
	
	Using this and the fact that $\norm{\cdot}_{\mb} \geq 0$ we can upper bound $\expec{\norm{\xi}_{\mb}\middle\vert \mcalF}$ as follows:
	\begin{align*}
		\expec{\norm{\xi}_{\mb}\middle\vert \mcalF} \leq \frac{1}{\prob{\mcalF}} \cdot \expec{\norm{\xi}_{\mb}}
		\leq \frac{c_1}{900}\cdot \sqrt{\lambda_1(\bv{B}^{-1})}
	\end{align*}
	Plugging this into~\eqref{eqn:expec-potential}, we obtain:
	\[
	\expec{G(\xhat)\middle\vert \mcalF} 
	\leq 
	\frac{1}{50} G(x) + \frac{2\expec{\norm{\xi}_{\mb}\middle\vert \mcalF}}{\sqrt{\lambda_1(\bv{B}^{-1})}}
	\leq \frac{1}{50}\cdot G(x) + \frac{2c_1}{900}.
	\]
	We can now finally bound $\expec{G(\xtilde)}$ as follows:
	\begin{align*}
		\expec{G(\xtilde)} &= \prob{\mcalF} \cdot \expec{G(\xhat)\middle\vert \mcalF} + \left(1-\prob{\mcalF}\right) G(x) \\
		&\leq \frac{9}{10} \left ( \frac{1}{50}\cdot G(x) + \frac{2c_1}{900}\right) + \frac{1}{10} G(x) 
		= \frac{3}{25} G(x) + \frac{2c_1}{1000}.
	\end{align*}
	This proves the theorem.
\end{proof}

\begin{corollary} [Relative Error Linear System Solvers]
	\label{cor:constant_factor_corollary} 
		For any unit vector $x$, we have:
		\begin{equation}
		\label{eq:constant_factor_corollary}
		\norm{\frac{1}{x^\top \mb x} x - \mb^{-1}x}_{\mb} \leq \alpha_1\sqrt{\lambda_1(\bv{B}^{-1})} \cdot G(x) =  \lamiBinv{1} \sqrt{\sum_{i\geq 2} \frac{\alpha_i^2}{\lamiBinv{i}}},
		\end{equation}
		where $x=\sum_i \alpha_i v_i$ is the decomposition of $x$ along $v_i$.
	Therefore, instantiating Theorem \ref{thm:powermethod-perturb} with $c_1 = \alpha_1 G(x)$ gives $\E [G(\xtilde)] \le \frac{4}{25} G(x)$ as long as:
	\begin{align*}
		\expec{\norm{\solve{x}-\mb^{-1}x}_{\mb}} \le  \frac{1}{1000} \norm{\frac{1}{\lambda-x^\top \bv{\Sigma}x} x - \mb^{-1}x}_{\mb}.
	\end{align*}
\end{corollary}
\begin{proof}
	Since $\mb$ is PSD we see that if we let $f(w) = \frac{1}{2} w^\top \mb w - w^\top x$, then the minimizer is $\mb^{-1} x$. Furthermore note that $\frac{1}{x^\top \mb x} = \argmin_\beta
	f(\beta x)$ and therefore
	\begin{align*}
	 \norm{\frac{1}{x^\top \mb x} x - \mb^{-1}x}_{\mb}^2
	&=  x^\top \mb^{-1} x - \frac{1}{x^\top \mb x}
	= 2 \left [f\left(\frac{x}{x^\top \mb x}\right) - f(\mb^{-1} x) \right] \\
	= & 2 \left [\min_\beta f(\beta x) - f(\mb^{-1} x) \right]
	\le 2 \left[f(\lamiBinv{1}x) - f(\mb^{-1} x) \right] \\
	= &\lamiBinv{1}^2 x^\top \mb x - 2 \lamiBinv{1} x^\top x + x^\top \mb^{-1} x \\
	= & \sum_{i=1}^d \abs{v_i^\top \mb^{\frac{1}{2}} x}^2(\lamiBinv{1}-\lamiBinv{i})^2
	\le \lamiBinv{1}^2\sum_{i\ge 2} \abs{v_i^\top \mb^{\frac{1}{2}} x}^2 \\
	= & \lamiBinv{1}^2\sum_{i\geq 2} \frac{\alpha_i^2}{\lamiBinv{i}},
	\end{align*}
	which proves \eqref{eq:constant_factor_corollary}.	
	
Consequently
\begin{align*}
\frac{c_1}{1000}\sqrt{\lambda_1(\bv{B}^{-1})} = \frac{1}{1000} \alpha_1 G(x)\sqrt{\lambda_1(\bv{B}^{-1})} \ge \frac{1}{1000} \norm{\frac{1}{x^\top \mb  x} x - \mb^{-1}x}_{\mb}
\end{align*}
which with Theorem~\ref{thm:powermethod-perturb} then completes the proof.
\end{proof}

\subsection{Initialization}
\label{sec:framework:init}

Theorem \ref{thm:powermethod-perturb} and Corollary \ref{cor:constant_factor_corollary} show that, given a good enough approximation to $v_1$, we can rapidly refine this approximation by applying the shifted-and-inverted power method. In this section, we cover initialization. That is, how to obtain a good enough approximation to apply these results.


We first give a simple bound on the quality of a randomly chosen start vector $x_0$. 

\begin{lemma}[Random Initialization Quality]\label{lem:init-random}
    Suppose $x \sim \mathcal{N}(0, \bv{I})$, and we initialize $x_0$ as $\frac{x}{\norm{x}_2}$, then with probability greater than $1- O\left (\frac{1}{d^{10}} \right )$, we have:
    \begin{equation*}
    	G(x_0) \le \sqrt{\kappa(\mb^{-1})} d^{10.5} \le 15 \frac{1}{\sqrt{\gap}} \cdot d^{10.5}
    \end{equation*}
    where $\kappa(\mb^{-1}) = \lambda_1(\mb^{-1}) / \lambda_d(\mb^{-1)}$.
\end{lemma}

\begin{proof}
	\begin{align*}
		G(x_0) =& G(x) = \frac{\norm{\bv{P}_{v_1^{\perp}}\left(x\right)}_{\mb}} {\norm{\bv{P}_{v_1}(x)}_{\mb}} = \frac{\sqrt{\norm{x}_{\mb}^{2}-\left(v_1^\top\mb^{1/2}x\right)^{2}}}{\abs{v_1^\top \mb^{1/2}x}}
		=\frac{\sqrt{\sum_{i\geq2} \frac{(v_i^\top x)^2}{\lamiBinv{i}}}}{\sqrt{\frac{(v_1^\top x)^2}{\lamiBinv{1}}}}, \\
		\le& \sqrt{\kappa(\mb^{-1})} \cdot \frac{\sqrt{\sum_{i\geq2}(v_i^\top x)^2}}{\abs{v_1^\top x}}
	\end{align*}
	Since $\{v_i^\top x\}_i$ are independent standard normal Gaussian variables.
	By standard concentration arguments, with probability greater than $1- e^{-\Omega(d)}$, we have $\sqrt{\sum_{i\geq2}(v_i^\top x)^2} = O\left (\sqrt{d} \right )$. Meanwhile, $v_1^\top x$ is just a one-dimensional standard Gaussian. It is easy to show $\Pr \left(\abs{v_1^\top x} \le \frac{1}{d^{10}} \right) = O\left (\frac{1}{d^{10}} \right)$, which finishes the proof.
\end{proof}

We now show that we can rapidly decrease our initial error to obtain the required $G(x) \le \frac{1}{\sqrt{10}}$ bound for Theorem \ref{thm:powermethod-perturb}.

\begin{theorem}[Approximate Shifted-and-Inverted Power Method -- Burn-In]\label{thm:init-offline}
	Suppose we initialize $x_0$ as in Lemma \ref{lem:init-random} and suppose we have access to a subroutine $\solve{\cdot}$ such that
	\begin{align*}
		\expec{\norm{\solve{x}-\mb^{-1}x}_{\mb}} \leq \frac{1}{3000 \kappa(\mb^{-1})d^{21}} \cdot \norm{\frac{1}{\lambda-\rayquot(x)} x - \mb^{-1}x}_{\mb}
	\end{align*}
	where $\kappa(\mb^{-1}) = \lambda_1(\mb^{-1}) / \lambda_d(\mb^{-1)}$.
	Then the following procedure,
	\begin{align*}
		x_{t} = \solve{x_{t-1}}/\norm{\solve{x_{t-1}}}
	\end{align*}
	after $T = O\left(\log d + \log \kappa(\mb^{-1}))\right)$ iterations  satisfies:
	\begin{align*}
		G(x_T) \leq \frac{1}{\sqrt{10}},
	\end{align*}
	with probability greater than $1- O(\frac{1}{d^{10}})$.
\end{theorem}

\begin{proof}
As before, we first bound the numerator and denominator of $G(\xhat)$ more carefully as follows:
\begin{align*}
	\begin{array}{lrl}
	\textbf{Numerator:}	&\norm{\bv{P}_{v_1^{\perp}}\left(\xhat \right)}_{\mb}
	&\leq \norm{\bv{P}_{v_1^{\perp}}\left(\mb^{-1} x \right)}_{\mb} + \norm{\bv{P}_{v_1^{\perp}}\left(\xi\right)}_\mb
	\leq \norm{\bv{P}_{v_1^{\perp}}\left(\mb^{-1} x \right)}_{\mb} + \norm{\xi}_{\mb} \\
	& &=\sqrt{\sum_{i\geq2}\left(v_{i}^{T}B^{-1/2}x\right)^{2}}+\norm{\xi}_{\mb} = \sqrt{\sum_{i\geq2} \alpha_{i}^{2}\lamiBinv{i}} + \norm{\xi}_{\mb}, \\
	\textbf{Denominator:} &\norm{\bv{P}_{v_1}\left(\xhat \right)}_{\mb}
	&\geq  \norm{\bv{P}_{v_1}\left(\mb^{-1} x \right)}_{\mb} - \norm{\bv{P}_{v_1}\left(\xi\right)}_\mb
	\geq \norm{\bv{P}_{v_1}\left(\mb^{-1} x \right)}_{\mb} - \norm{\xi}_{\mb} \\
	&&=\abs{v_{i}^{T}B^{-1/2}x}-\norm{\xi}_{\mb} = \alpha_1\sqrt{\lamiBinv{1}} - \norm{\xi}_{\mb}
	\end{array}
\end{align*}

We now use the above estimates to bound $G(\xhat)$.
\begin{align*}
	G(\xhat)  &\leq \frac{\sqrt{\sum_{i\geq2} \alpha_{i}^{2}\lamiBinv{i}} + \norm{\xi}_{\mb}}{\alpha_1\sqrt{\lamiBinv{1}} - \norm{\xi}_{\mb}} 
	\leq \frac{\lamiBinv{2}\sqrt{\sum_{i\geq2} \frac{\alpha_{i}^{2}}{\lamiBinv{i}}} + \norm{\xi}_{\mb}}{\lamiBinv{1}\sqrt{\frac{\alpha_1^2}{\lamiBinv{1}}} - \norm{\xi}_{\mb}} \\
	&= G(x) \frac{\lamiBinv{2} + \norm{\xi}_{\mb}/\sqrt{\sum_{i\geq2} \frac{\alpha_{i}^{2}}{\lamiBinv{i}}}}{\lamiBinv{1} - \norm{\xi}_{\mb}/\sqrt{\frac{\alpha_1^2}{\lamiBinv{1}}}}
\end{align*}
By Lemma \ref{lem:init-random}, we know with at least probability $1- O(\frac{1}{d^{10}})$, 
we have $G(x_0) \le \sqrt{\kappa(\mb^{-1})} d^{10.5}$.

Conditioned on high probability result of $G(x_0)$, we now use induction to prove $G(x_t) \le G(x_0)$. It trivially holds for $t=0$.
Suppose we now have $G(x) \le G(x_0)$, then by the condition in Theorem \ref{thm:init-offline} and Markov inequality, we know with probability greater than $1-\frac{1}{100\sqrt{\kappa(\mb^{-1})} d^{10.5}}$ we have: 
\begin{align*}
	\norm{\xi}_B \le &\frac{1}{30\sqrt{\kappa(\mb^{-1})} d^{10.5}} \cdot \norm{\frac{1}{\lambda-\rayquot(x)} x - \mb^{-1}x}_{\mb}\\
	\le & \frac{1}{30} \cdot \norm{\frac{1}{\lambda-\rayquot(x)} x - \mb^{-1}x}_{\mb} \min \left \{1, \frac{1}{G(x_0)} \right \} \\
	\le &\frac{1}{30} \cdot \norm{\frac{1}{\lambda-\rayquot(x)} x - \mb^{-1}x}_{\mb} \min \left \{1, \frac{1}{G(x)} \right \} \\
	\le & \frac{\lamiBinv{1}-\lamiBinv{2}}{4}\min\left\{\sqrt{\sum_{i\geq2} \frac{\alpha_{i}^{2}}{\lamiBinv{i}}}, \sqrt{\frac{\alpha_1^2}{\lamiBinv{1}}}\right\}
\end{align*}
The last inequality uses Corollary~\ref{cor:constant_factor_corollary} with the fact that $\lamiBinv{2} \le \frac{1}{100}\lamiBinv{1}$. Therefore, we have:
We will have:
\begin{equation*}
	G(\xhat) \le  \frac{\lamiBinv{1} + 3\lamiBinv{2}}{3\lamiBinv{1}+\lamiBinv{2}} \times G(x)
	\le \frac{1}{2}G(x)
\end{equation*}
This finishes the proof of induction.

Finally, by union bound, we know with probability greater than $1- O(\frac{1}{d^{10}})$ in $T = O(\log d + \log \kappa(\mb^{-1}))$ steps, 
we have:
\begin{equation*}
	G(x_T) \le  \frac{1}{2^T} G(x_0) \le \frac{1}{\sqrt{10}}
\end{equation*}

\end{proof}

%% file: offline.tex
\section{Offline Eigenvector Computation}\label{sec:offline}

In this section we show how to instantiate the framework of Section \ref{framework} in order to compute an approximate top eigenvector in the offline setting. As discussed, in the offline setting we can trivially compute the Rayleigh quotient of a vector in $\nnz(\ma)$ time as we have explicit access to $\ma^\top \ma$. Consequently the bulk of our work in this section is to show how we can solve linear systems in $\mb$ efficiently in expectation, allowing us to apply Corollary \ref{cor:constant_factor_corollary} of Theorem \ref{thm:powermethod-perturb}.

In Section~\ref{sec:offline:svrg} we first show how Stochastic Variance Reduced Gradient (SVRG) \cite {johnson2013accelerating} can be adapted to solve linear systems of the form $\mb x = b$.
If we wanted, for example, to solve a linear system in a positive definite matrix like $\bv{A}^\top  \bv A$, we would optimize the objective function $f(x) = \frac{1}{2}x^\top \bv{A}^\top  \bv A x - b^\top x$. This function can be written as the sum of $n$ \emph{convex components}, $\psi_i(x) = \frac{1}{2} x^\top \left (a_i a_i^\top \right )x + \frac{1}{n}b^\top x$. In each iteration of traditional gradient descent, one computes the full gradient of $f(x_i)$ and takes a step in that direction. In stochastic gradient methods, at each iteration, a single component is sampled, and the step direction is based only on the gradient of the sampled component. Hence, we avoid a full gradient computation at each iteration, leading to runtime gains.

Unfortunately, while we have access to the rows of $\bv{A}$ and so can solve systems in $\bv A^\top \bv A$, it is less clear how to solve systems in $\bv{B} = \lambda \bv I - \bv A^\top \bv A$. To do this, we will split our function into components of the form $\psi_i(x) = \frac{1}{2} x^\top \left (w_i \bv I - a_i a_i^\top \right )x + \frac{1}{n}b^\top x$ for some set of weights $w_i$ with $\sum_{i \in [n]} w_i = \lambda$.

Importantly, $w_i \bv I - a_i a_i^\top$ may not be positive semidefinite. That is, we are minimizing a sum of functions which is convex, but consists of non-convex components. While recent results for minimizing such functions could be applied directly \cite{shalev2015sdca,csiba2015primal} here we show how to obtain stronger results by using a more general form of SVRG and analyzing the specific properties of our function (i.e. the variance).

Our analysis shows that we can make constant factor progress in solving linear systems in $\bv{B}$ in time $O\left (\nnz(\bv A) + \frac{d\nrank(\bv A)}{\gap^2} \right )$. If $\frac{d\nrank(\bv A)}{\gap^2} \le \nnz(\bv A)$ this gives a runtime proportional to the input size -- the best we could hope for. 
If not, we show in Section~\ref{sec:offline:acceleration}  that it is possible to \emph{accelerate} our system solver, achieving runtime  $\tilde O \left (\frac{\sqrt{\nnz(\bv A)d\nrank(\bv A)}}{\gap} \right )$. This result uses the work of \cite{frostig2015regularizing, lin2015catalyst} on accelerated approximate proximal point algorithms. 

With our solvers in place, in
Section~\ref{sec:offline:results} we pull our results together, showing how to use these solvers in the framework of Section \ref{framework} to give faster running times for offline eigenvector computation.

\subsection{SVRG Based Solver}
\label{sec:offline:svrg}

Here we provide a sampling based algorithm for solving linear systems in $\mb$. In particular we provide an algorithm for solving the more general problem where we are given a strongly convex function that is a sum of possibly non-convex functions that obey smoothness properties. We provide a general result on bounding the progress of an algorithm that solves such a problem by non-uniform sampling in Theorem~\ref{lem:svrg-nonconv} and then in the remainder of this section we show how to bound the requisite quantities for solving linear systems in $\mb$.

\newcommand{\avgsmooth}{\overline{S}}

\begin{theorem}[SVRG for Sums of Non-Convex Functions]
\label{lem:svrg-nonconv} 
Consider a set of functions, $\{\psi_{1}, \psi_2,...\psi_n \}$, each mapping $\R^{d}\rightarrow\R$.
Let $f(x) = \sum_{i}\psi_{i}(x)$ and let $\opt{x} \eqdef \argmin_{x \in \R^{d}} f(x)$. Suppose we have a probability distribution $p$ on $[n]$, 
and that starting from some initial point $x_{0} \in \R^d$ in each iteration $k$ we pick $i_{k} \in [n]$ independently with probability $p_{i_k}$ and let 
\[
x_{k+1}
:= 
x_{k}
-\frac{\eta}{p_{i}} 
\left(\grad\psi_{i}(x_k)-\grad\psi_{i}(x_0 )\right) + \eta \grad f(x_0)
\]
for some $\eta$. If $f$ is $\mu$-strongly convex and if for all $x \in \R^d$
we have
\begin{equation}
\label{eq:svrg-nonconvex-avgsmooth}
\sum_{i \in [n]} \frac{1}{p_{i}} \norm{\grad \psi_{i}(x) - \grad\psi_{i}(\opt{x})}_{2}^{2} 
\leq 
2 \avgsmooth \left[f(x) - f(\opt{x}) \right],
\end{equation}
where $\avgsmooth$ is a variance parameter,
then for all $m \geq 1$ we have
\[
\E \left[\frac{1}{m} \sum_{k\in[m]} f(x_{k}) - f(\opt{x})\right] \leq \frac{1}{1-2\eta\bar{S}} 
\left[\frac{1}{\mu\eta m}+2\eta\avgsmooth \right] \cdot 
\left[f(x_{0})-f(\opt{x})\right]
\]
Consequently, if we pick $\eta$ to be a sufficiently small multiple of $1/\bar{S}$
then when $m = O(\overline{S} / \mu)$ we can decrease the error by
a constant multiplicative factor in expectation.
\end{theorem}

\begin{proof}
We first note that $\E_{i_{k}}[x_{k+1} - x_{k}] = \eta \grad f(x_{k})$. This is, in each iteration, in expectation, we make a step in the direction of the gradient. Using this fact we have:
\begin{align*}
\E_{i_{k}} \norm{x_{k+1} - \opt{x}}_{2}^{2} &= \E_{i_{k}} \norm{(x_{k+1} -x_k) + (x_k - \opt{x}) }_{2}^{2}  \\
&= \norm{x_k - \opt{x} }_{2}^{2} - 2 \E_{i_{k}}(x_{k+1} -x_k)^\top(x_k - \opt{x}) + \E_{i_{k}}\norm{x_{k+1} - x_k }_{2}^{2} \\
&=\norm{x_k - \opt{x} }_{2}^{2} - 2 \eta \grad f(x_k)^{\top} \left(x_{k} - \opt{x} \right)\\
&+ \sum_{i \in [n]} \eta^{2} p_{i} \normFull{
	\frac{1}{p_{i}} 
	\left(\grad \psi_{i}(x_{k}) - \grad\psi_{i}(x_{0}) \right)
	+ \grad f(x_{0})}_{2}^{2}
\end{align*}
We now apply the fact that $\norm{x + y}_2^2 \leq 2\norm{x}_2^2 + 2\norm{y}_2^2$ to give:
\begin{align*}
\sum_{i \in [n]}
& p_i \normFull{\frac{1}{p_i} 
\left(\grad \psi_i (x_k) - \grad \psi_i (x_0)\right) + \grad f(x_0)
}_2^2
\\
&
\leq 
\sum_{i \in [n]}
2 p_i \normFull{\frac{1}{p_i} 
	\left(\grad \psi_i (x_k) - \grad \psi_i (\opt{x}) \right)
}_2^2
+
\sum_{i \in [n]}
2 p_i \normFull{\frac{1}{p_i} 
	\left(\grad \psi_i (x_0) - \grad \psi_i(\opt{x}) \right) - \grad f(x_0)
}_2^2.
\end{align*}
Then, using that $\grad f(\opt{x}) = 0$ by optimality, that $\E\norm{x-\E x}_{2}^{2} \leq \E\norm{x}_{2}^{2}$, and  \eqref{eq:svrg-nonconvex-avgsmooth} we have:
\begin{align*}
\sum_{i \in [n]}
& p_i \normFull{\frac{1}{p_i} 
\left(\grad \psi_i (x_k) - \grad \psi_i (x_0)\right) + \grad f(x_0)
}_2^2
\\
&
\le
\sum_{i \in [n]}
\frac{2}{p_i}
 \normFull{\grad \psi_i (x_k) - \grad \psi_i (\opt{x})
}_2^2
+
\sum_{i \in [n]}
2 p_i \normFull{\frac{1}{p_i} 
	\left(\grad \psi_i (x_0) - \grad \psi_i(\opt{x})) - (\grad f(x_0) - \grad f(\opt{x})\right)
}_2^2
\\
&
\leq
\sum_{i \in [n]}
\frac{2}{p_i}
\normFull{\grad \psi_i (x_k) - \grad \psi_i (\opt{x})
}_2^2
+
\sum_{i \in [n]}
2 p_i \normFull{\frac{1}{p_i} 
	\grad \psi_i (x_0) - \grad \psi_i(\opt{x}))
}_2^2
\\
&
\leq 4 \avgsmooth
\left[ f(x_{k})-f(\opt{x}) + f(x_{0}) - f(\opt{x}) \right]
\end{align*}

Since $f(\opt{x}) - f(x_k) \geq \grad f(x_k)^\top (\opt{x} - x_k)$ by the convexity of $f$, these inequalities imply
\begin{align*}
\E_{i_{k}}\norm{x_{k+1} - \opt{x}}_{2}^{2} 
& 
\leq\norm{x_{k} - \opt{x}}_{2}^{2} 
- 2 \eta \left[f(x_{k}) - f(\opt{x})\right] 
+ 4 \eta^{2} \avgsmooth \left[f(x_{k}) - f(\opt{x}) + f(x_{0}) - f(\opt{x})\right] 
\\
& 
= \norm{x_{k} - \opt{x}}_{2}^{2}-2\eta(1-2\eta S)\left(f(x_{k})-f(\opt{x})\right)+4\eta^{2}\bar{S}\left(f(x_{0})-f(\opt{x})\right)
\end{align*}
Rearranging, we have:
\begin{align*}
2\eta(1-2\eta S)\left(f(x_{k})-f(\opt{x})\right) \le \norm{x_{k} - \opt{x}}_{2}^{2} - \E_{i_{k}}\norm{x_{k+1} - \opt{x}}_{2}^{2} + 4\eta^{2}\bar{S}\left(f(x_{0})-f(\opt{x})\right).
\end{align*}
And summing over all iterations and taking expectations we have:
\begin{align*}
\E \left [2\eta(1-2\eta \bar{S})\sum_{k\in[m]}f(x_{k})-f(\opt{x}) \right ] \le \norm{x_0 - \opt{x}}_2^2 + 4m\eta^{2}\bar{S}\left[f(x_{0})-f(\opt{x})\right].
\end{align*}
Finally, we use that
by strong convexity, $\norm{x_{0}-\opt{x}}_{2}^{2}\leq\frac{2}{\mu}\left(f(x_{0})-f(\opt{x})\right)$ to obtain:
	\[
	\E \left [2\eta(1-2\eta \bar{S})\sum_{k\in[m]}f(x_{k})-f(\opt{x})\right ]\leq \frac{2}{\mu}\left[f(x_{0})-f(\opt{x})\right]+4m\eta^{2}\bar{S}\left[f(x_{0})-f(\opt{x})\right]
	\]
	and thus
	\[
	\E \left [\frac{1}{m}\sum_{k\in[m]}f(x_{k})-f(\opt{x}) \right ]\leq\frac{1}{1-2\eta\bar{S}}\left[\frac{1}{\mu\eta m}+2\eta\bar{S}\right]\cdot\left[f(x_{0})-f(\opt{x})\right]
	\]
	
\end{proof}

Theorem \ref{lem:svrg-nonconv} immediately yields a solver for $\bv{B}x = b$. Finding the minimum norm solution to this system is equivalent to minimizing $f(x) = \frac{1}{2} x^\top \bv{B} x - b^\top x$. If we take the common approach of applying a smoothness bound for each $\psi_i$ along with a strong convexity bound on $f(x)$ we obtain:
\begin{lemma}[Simple Variance Bound for SVRG]\label{simple_variance_bound}
Let
\begin{align*}
	\psi_i(x) \defeq \frac{1}{2} x^\top \left(\frac{\lambda\norm{a_i}_2^2}{\norm{\bv A}_F^2} \mI - a_i a_i^\top \right) x 
	- \frac{1}{n}b^\top x
\end{align*}
so we have $\sum_{i\in[n]} \psi_i(x) = f(x) = \frac{1}{2} x^\top \bv{B} x - b^\top x$.
Setting $p_i = \frac{\norm{a_i}_2^2}{\norm{\bv A}_F^2}$ for all $i$, we have 
\begin{align*}
\sum_{i \in [n]} \frac{1}{p_i} \norm{\grad \psi_i(x)-\grad \psi_i(\opt{x})}_2^2 = O \left (
\frac{\norm{\bv A}_F^4}{\lambda - \lambda_1}
\left[f(x) - f(\opt{x})\right] \right )
\end{align*}
\end{lemma}

\begin{proof}
	We first compute, for all $i \in [n]$ 
	\begin{align}\label{basic_gradient_computation}
	\grad \psi_i(x) =
	\left(\frac{\lambda \norm{a_i}_2^2}{\normFro{\ma}^2} \mI 
	- a_i a_i^\top \right) x 
	- \frac{1}{n}b.
	\end{align}
	We have that each $\psi_i$ is $\frac{\lambda \norm{a_i}_2^2}{\normFro{\ma}^2} + \norm{a_i}^2$ smooth with respect to $\norm{\cdot}_2$. Specifically,
\begin{align*}
\norm{\grad \psi_i(x)-\grad \psi_i(\opt{x})}_2 &= \norm {\left (\frac{\lambda \norm{a_i}_2^2}{\normFro{\ma}^2} \mI 
	- a_i a_i^\top \right) (x-\opt{x}) }_2 \\
	&\le \left (\frac{\lambda \norm{a_i}_2^2}{\normFro{\ma}^2} + \norm{a_i}^2 \right ) \norm{x-\opt{x}}_2.
\end{align*}

Additionally, 
$f(x)$ is $\lambda_d(\bv B) = \lambda - \lambda_1$ strongly convex so we have $\norm{x - \opt{x}}_2^2 \leq \frac{2}{\lambda - \lambda_1} \left[f(x) - f(\opt{x})\right]$ and putting all this together we have
\begin{align*}
\sum_{i \in [n]} \frac{1}{p_i} \norm{\grad \psi_i(x) - \grad \psi_i (\opt{x})}_2^2
&\leq
\sum_{i \in [n]} \frac{\norm{\bv A}_F^2}{\norm{a_i}_2^2}
\cdot
\norm{a_i}_2^4 \left(\frac{\lambda}{\norm{\bv A}_F^2} + 1\right)^2
\cdot \frac{2}{\lambda - \lambda_1}
\left[f(x) - f(\opt{x})\right]
\\
&= O \left (
\frac{\norm{\bv A}_F^4}{\lambda - \lambda_1}
\left[f(x) - f(\opt{x})\right] \right )
\\
\end{align*}
where the last step uses that $\lambda \le 2\lambda_1 \le 2\norm{\bv A}_F^2$ so $\frac{\lambda}{\norm{\bv A}_F^2} \le 2$.
\end{proof}

Assuming that $\lambda = (1+c \cdot \gap)\lambda_1$ for some constant $c$, the above bound means that we can make constant progress on our linear system by setting $m = O(\avgsmooth/\mu) = O\left(\frac{\norm{\bv A}^4_F}{(\lambda-\lambda_1)^2}\right)= O\left ( \frac{\nrank(\bv{A})^2}{\gap^2} \right )$. This dependence on stable rank matches the dependence given in \cite{shamir2015stochastic} (see discussion in Section \ref{previous_work_offline}), however we can show that it is suboptimal. 
We show to improve the bound to $O\left ( \frac{\nrank(\bv{A})}{\gap^2} \right )$ by using a better variance analysis. Instead of bounding each $\norm{\grad \psi_i(x) - \grad \psi_i (\opt{x})}_2^2$ term using the smoothness of $\psi_i$, we more carefully bound the sum of these terms.

\begin{lemma}(Improved Variance Bound for SVRG)\label{variance_lemma} For $i \in [n]$ let
\begin{align*}
	\psi_i(x) \defeq \frac{1}{2} x^\top \left(\frac{\lambda\norm{a_i}_2^2}{\norm{\bv A}_F^2} \mI - a_i a_i^\top \right) x 
	- \frac{1}{n}b^\top x
\end{align*}
so we have $\sum_{i\in[n]} \psi_i(x) = f(x) = \frac{1}{2} x^\top \bv{B} x - b^\top x$.
Setting $p_i = \frac{\norm{a_i}_2^2}{\norm{\bv A}_F^2}$ for all $i$, we have for all $x$
	\[
	\sum_{i \in [n]}
	\frac{1}{p_i} \normFull{\grad \psi_i(x) - \grad \psi_i(\opt{x})}_2^2
	\leq 
	\frac{4\lambda_1 \normFro{\ma}^2}{\lambda - \lambda_1} \cdot \left[f(x) - f(\opt{x}) \right].
	\]
\end{lemma}

\begin{proof}
Using the gradient computation in \eqref{basic_gradient_computation} we have
	\begin{align}
	\sum_{i \in [n]}
	\frac{1}{p_i} \normFull{\grad \psi_i(x) - \grad \psi_i(\opt{x})}_2^2 &= \sum_{i \in [n]}
	\frac{\normFro{\ma}^2}{\norm{a_i}_2^2}
	\normFull{\left( \frac{\lambda\norm{a_i}_2^2}{\normFro{\ma}^2} \mI 
		- a_i a_i^\top \right) (x - \opt{x})}_2^2\nonumber\\
	&=
	\sum_{i \in [n]} 
	\frac{\lambda^2 \norm{a_i}_2^2}{\normFro{\ma}^2}
	\normFull{x - \opt{x}}^2_2
	- 2 \sum_{i \in [n]} \lambda \normFull{x - \opt{x}}_{a_i a_i^\top}^2
	\nonumber\\&+ \sum_{i \in [n]} \frac{\normFro{\ma}^2}{\norm{a_i}^2} \normFull{x - \opt{x}}_{\norm{a_i}_2^2 a_i a_i^\top}^2
	\nonumber\\
	&=
	\lambda^2 \normFull{x - \opt{x}}_2^2
	- 2 \lambda \normFull{x - \opt{x}}_{\mSigma}^2
	+ \normFro{\ma}^2 \normFull{x - \opt{x}}_{\mSigma}^2.\nonumber\\
		&\le
	\lambda \normFull{x - \opt{x}}_{\bv B}^2
	+ \normFro{\ma}^2 \normFull{x - \opt{x}}_{\mSigma}^2.\label{broken_up_norms}
	\end{align}
	
	Now since 
	\[
	\mSigma \preceq \lambda_1 \mI \preceq \frac{\lambda_1}{\lambda - \lambda_1} \mb 
	\]
	we have
	\begin{align*}
	\sum_{i \in [n]}
	\frac{1}{p_i} \normFull{\grad \psi_i(x) - \grad \psi_i(\opt{x})}_2^2  &\le \left (\frac{ \lambda(\lambda-\lambda_1) +\normFro{\ma}^2\cdot \lambda_1}{\lambda - \lambda_1} \right )\normFull{x - \opt{x}}_{\mb}^2\\
	&\le \left (\frac{2\normFro{\ma}^2 \lambda_1}{\lambda - \lambda_1} \right )\normFull{x - \opt{x}}_{\mb}^2
	\end{align*}
	where in the last inequality we just coarsely bound $\lambda(\lambda-\lambda_1) \le \lambda_1\norm{\bv A}_F^2$.
	Now since $\bv{B}$ is full rank, $\bv{B}\opt{x} = b$, we can compute:
	\begin{align}\label{norm_to_function_error_conversion}
	\norm{x - \opt{x}}_{\mb}^2 = x^\top \mb x - 2b^\top x + b^\top \opt{x} = 2[f(x) - f(\opt{x})].
	\end{align}
	The result follows.
\end{proof}

Plugging the bound in Lemma \ref{variance_lemma} into Theorem \ref{lem:svrg-nonconv} we have:
\begin{theorem}(Offline SVRG-Based Solver)\label{offline_solver}
Let $\avgsmooth =\frac{2\lambda_1 \normFro{\mSigma}^2}{\lambda - \lambda_1}$, $\mu =\lambda-\lambda_1$. The iterative procedure described in Theorem \ref{lem:svrg-nonconv} with $f(x) = \frac{1}{2}x^\top \bv{B} x - b^\top x$, $\psi_i(x) = \frac{1}{2} x^\top \left(\frac{\lambda\norm{a_i}_2^2}{\norm{\bv \Sigma}_F^2} \mI - a_i a_i^\top \right) x - b^\top x$, $p_i = \frac{\norm{a_i}_2^2}{\norm{\bv \Sigma}_F^2}$, $\eta = 1/(8\avgsmooth)$ and $m$ chosen uniformly at random from $[64 \avgsmooth/\mu]$ returns a vector $x_m$ such that
\begin{align*}
\E \norm{x_m - \opt{x}}_\mb^2 \le \frac{1}{2} \norm{x_0-\opt{x}}_\mb^2.
\end{align*}
Further, assuming $\left (1 + \frac{\gap}{150} \right ) \lambda_1< \lambda \le \left (1 + \frac{\gap}{100} \right ) \lambda_1$, this procedure runs in time $O\left (\nnz(\bv A) + \frac{d\cdot \nrank(\bv{A})}{\gap^2} \right )$.
\end{theorem}
\begin{proof}
Lemma \ref{variance_lemma} tells us that \[
	\sum_{i \in [n]}
	\frac{1}{p_i} \normFull{\grad \psi_i(x) - \grad \psi_i(\opt{x})}_2^2
	\leq 
	2 \avgsmooth \left[f(x) - f(\opt{x}) \right].
	\]
Further $f(x) = \frac{1}{2}x^\top \bv{B} x - b^\top x$ is $\lambda_d(\bv B)$-strongly  convex and $\lambda_d(\bv B) = \lambda-\lambda_1 = \mu$.
Plugging this into Theorem \ref{lem:svrg-nonconv} and using \eqref{norm_to_function_error_conversion} which shows $\norm{x-\opt{x}}_\mb^2 = 2[f(x)-f(\opt{x})]$ we have, for $m$ chosen uniformly from $[64 \avgsmooth/ \mu]$: 
\begin{align*}
\E \left[\frac{1}{64 \avgsmooth/ \mu} \sum_{k\in[64 \avgsmooth/ \mu]} f(x_{k}) - f(\opt{x})\right] &\leq 4/3 \cdot 
\left[1/8+1/8 \right] \cdot 
\left[f(x_{0})-f(\opt{x})\right]\\
\E \left [f(x_m) - f(\opt{x}) \right ] &\le \frac{1}{2}\left[f(x_{0})-f(\opt{x})\right] \\
\E \norm{x_m - \opt{x}}_\mb^2 &\le \frac{1}{2}  \norm{x_0-x^{opt}}_\mb^2.
\end{align*}

The procedure requires $O\left(\nnz(\bv A)\right)$ time to initially compute $\grad f(x_0)$, along with each $p_i$ and the step size $\eta$ which depend on $\norm{\bv A}_F^2$ and the row norms of $\bv{A}$. Each iteration then just requires $O(d)$ time to compute $\grad \psi_i(\cdot)$ and perform the necessary vector operations. Since there are at most $[64 \avgsmooth/ \mu] = O\left ( \frac{\lambda_1\norm{\bv A}_F^2}{(\lambda-\lambda_1)^2} \right )$ iterations, our total runtime is
\begin{align*}
O \left (\nnz(\bv{A})  + d \cdot \frac{\lambda_1\norm{\bv A}_F^2}{(\lambda-\lambda_1)^2} \right ) = O\left (\nnz(\bv A) + \frac{d\cdot \nrank(\bv{A})}{\gap^2} \right ).
\end{align*}
Note that if our matrix is uniformly sparse - i.e. all rows have sparsity at most $d_s$, then the runtime is actually at most $O\left (\nnz(\bv A) + \frac{d_s \cdot \nrank(\bv{A})}{\gap^2} \right )$.
\end{proof}

\subsection{Accelerated Solver}
\label{sec:offline:acceleration}

Theorem \ref{offline_solver} gives a linear solver for $\bv{B}$ that makes progress in expectation and which we can plug into Theorems \ref{thm:powermethod-perturb} and \ref{thm:init-offline}. However, we first show that the runtime in Theorem \ref{offline_solver} can be accelerated in some cases. We apply a result of \cite{frostig2015regularizing}, which shows that, given a solver for a regularized version of a convex function $f(x)$, we can produce a fast solver for $f(x)$ itself. Specifically:

\begin{lemma}[Theorem 1.1 of \cite{frostig2015regularizing}]\label{acceleration_primitive}
Let $f(x)$ be a $\mu$-strongly convex function and let $\opt{x} \eqdef \argmin_{x\in\mathbb{R}^d} f(x)$. For any $\gamma > 0$ and any $x_0 \in \mathbb{R}^d$, let $f_{\gamma,x_0}(x) \eqdef f(x) + \frac{\gamma}{2} \norm{x-x_0}_2^2$. Let $\opt{x}_{\gamma,x_0}  \eqdef \argmin_{x\in\mathbb{R}^d} f_{\gamma,x_0} (x)$.

Suppose that, for all $x_0 \in \mathbb{R}^d$, $c > 0$, $\gamma > 0$, we can compute a point $x_c$ such that
\begin{align*}
\E f_{\gamma,x_0}(x_c) - f_{\gamma,x_0}(\opt{x}_{\gamma,x_0} ) \le \frac{1}{c} \left [f_{\gamma,x_0} - f_{\gamma,x_0}(\opt{x}_{\gamma,x_0} ) \right ]
\end{align*}

in time $\mathcal{T}_c$. Then given any $x_0$, $c > 0$, $\gamma > 2 \mu$, we can compute $x_1$ such that
\begin{align*}
\E f(x_1) - f(\opt{x}) \le \frac{1}{c} \left [f(x_0) - f(\opt{x}) \right ]
\end{align*}
in time $O\left(\mathcal{T}_{4\left (\frac{2\gamma + \mu}{\mu} \right )^{3/2}} \sqrt{\lceil \gamma/\mu \rceil} \log c \right ).$
\end{lemma}

We first give a new variance bound on solving systems in $\bv{B}$ when a regularizer is used. The proof of this bound is very close to the proof given for the unregularized problem in Lemma \ref{variance_lemma}.
\begin{lemma}\label{regularized_variance_lemma} For $i \in [n]$ let
\begin{align*}
	\psi_i(x) \defeq \frac{1}{2} x^\top \left(\frac{\lambda\norm{a_i}_2^2}{\norm{\bv A}_F^2} \mI - a_i a_i^\top \right) x 
	- \frac{1}{n}b^\top x + \frac{\gamma\norm{a_i}_2^2}{2\norm{\bv A}_F^2} \norm{x - x_0}_2^2
\end{align*}
so we have $\sum_{i\in[n]} \psi_i(x) = f_{\gamma,x_0} (x) = \frac{1}{2} x^\top \bv{B} x - b^\top x + \frac{\gamma}{2}\norm{x - x_0}_2^2$.
Setting $p_i = \frac{\norm{a_i}_2^2}{\norm{\bv A}_F^2}$ for all $i$, we have for all $x$
	\[
	\sum_{i \in [n]}
	\frac{1}{p_i} \normFull{\grad \psi_i(x) - \grad \psi_i(\opt{x}_{\gamma,x_0} )}_2^2
	\leq 
	\left (  \frac{\gamma^2 + 12\lambda_1\norm{\bv A}_F^2}{\lambda-\lambda_1+\gamma}\right ) \left [f_{\gamma,x_0}(x)-f_{\gamma,x_0}(\opt{x}_{\gamma,x_0}) \right ]
	\]
\end{lemma}

\begin{proof}
	We have for all $i \in [n]$ 
	\begin{align}\label{regularized_gradient_computation}
	\grad \psi_i(x) =
	\left(\frac{\lambda \norm{a_i}_2^2}{\normFro{\ma}^2} \mI 
	- a_i a_i^\top \right) x 
	- \frac{1}{n}b + \frac{\gamma\norm{a_i}_2^2}{2\norm{\bv A}_F^2} (x-2x_0)
	\end{align}
Plugging this in we have:
	\begin{align*}
	\sum_{i \in [n]}
	\frac{1}{p_i} \normFull{\grad \psi_i(x) - \grad \psi_i(\opt{x}_{\gamma,x_0} )}_2^2 &= \sum_{i \in [n]}
	\frac{\normFro{\ma}^2}{\norm{a_i}_2^2}
	\normFull{\left( \frac{\lambda\norm{a_i}_2^2}{\normFro{\ma}^2} \mI 
		- a_i a_i^\top \right) (x - \opt{x}_{\gamma,x_0} ) + \frac{\gamma\norm{a_i}_2^2}{2\norm{\bv A}_F^2} (x - \opt{x}_{\gamma,x_0} )}_2^2
	\end{align*}
	For simplicity we now just use the fact that  $\norm{x+y}_2^2 \le 2\norm{x}_2^2 + 2\norm{y}_2^2$ and apply our bound from equation \eqref{broken_up_norms} to obtain:
	\begin{align*}
	\sum_{i \in [n]}
	\frac{1}{p_i} \normFull{\grad \psi_i(x) - \grad \psi_i(\opt{x}_{\gamma,x_0} )}_2^2  &\le 
	2\lambda^2 \normFull{x - \opt{x}_{\gamma,x_0}}_2^2
	- 4 \lambda \normFull{x - \opt{x}_{\gamma,x_0}}_{\mSigma}^2
	+ 2\normFro{\mSigma}^2 \normFull{x - \opt{x}_{\gamma,x_0}}_{\mSigma}^2 \\&+ 2\sum_{i \in [n]}\frac{\norm{a_i}_2^2}{\normFro{\ma}^2} \frac{\gamma^2}{4} \norm{x-\opt{x}_{\gamma,x_0} }_2^2\\
	&\le 
	 \left (2\lambda^2 +\gamma^2/2 + 2\lambda_1\norm{\bv A}_F^2 - 4\lambda_1\lambda \right ) \normFull{x - \opt{x}_{\gamma,x_0} }_{2}^2\\
	&\le 
	 \left (\gamma^2/2 + 6\lambda_1\norm{\bv A}_F^2 \right ) \normFull{x - \opt{x}_{\gamma,x_0} }_{2}^2
	\end{align*}
	
	Now, $f_{\gamma,x_0}(\cdot)$ is $\lambda-\lambda_1+\gamma$ strongly convex, so 
	\begin{align*}
	\normFull{x - \opt{x}_{\gamma,x_0} }_{2}^2 \le \frac{2}{\lambda-\lambda_1+\gamma} \left [f_{\gamma,x_0}(x)-f_{\gamma,x_0}(\opt{x}_{\gamma,x_0}) \right ].
	\end{align*}
	So overall we have:
	\begin{align*}
	\sum_{i \in [n]} \frac{1}{p_i} \normFull{\grad \psi_i(x) - \grad \psi_i(\opt{x}_{\gamma,x_0} )}_2^2  &\le \left (  \frac{\gamma^2 + 12\lambda_1\norm{\bv A}_F^2}{\lambda-\lambda_1+\gamma}\right ) \left [f_{\gamma,x_0}(x)-f_{\gamma,x_0}(\opt{x}_{\gamma,x_0}) \right ]
	\end{align*}
\end{proof}

We can now use this variance bound to obtain an accelerated solver for $\bv{B}$. We assume $\nnz(\bv{A}) \le \frac{d\nrank(\bv A)}{\gap^2}$, as otherwise, the unaccelerated solver in Theorem \ref{offline_solver} runs in $O(\nnz(\bv{A}))$ time and cannot be accelerated further.
\begin{theorem}[Accelerated SVRG-Based Solver]\label{accelerated_offline_solver}
Assuming $\left (1 + \frac{\gap}{150} \right ) \lambda_1< \lambda \le \left (1 + \frac{\gap}{100} \right ) \lambda_1$ and $nnz(\bv{A}) \le \frac{d\nrank(\bv A)}{\gap^2}$, applying the iterative procedure described in Theorem \ref{lem:svrg-nonconv} along with the acceleration given by Lemma \ref{acceleration_primitive} gives a solver that returns $x$ with
\begin{align*}
\E \norm{x - \opt{x}}_\mb^2 \le \frac{1}{2} \norm{x_0-\opt{x}}_\mb^2.
\end{align*}
in time $O\left ( \frac{\nnz(\bv A)^{3/4} (d \nrank(\bv A))^{1/4}}{\sqrt{\gap}} \cdot \log \left (\frac{d}{\gap}\right)\right)$.
\end{theorem}
\begin{proof}
Following Theorem \ref{offline_solver}, the variance bound of Lemma \ref{regularized_variance_lemma} means that we can make constant progress in minimizing $f_{\gamma,x_0}(x)$ in $O\left (\nnz(\bv{A}) + d m \right)$ time where $m = O \left ( \frac{\gamma^2 + 12\lambda_1\norm{\bv\Sigma}_F^2}{(\lambda-\lambda_1+\gamma)^2}\right )$. So, for $\gamma \ge 2(\lambda-\lambda_1)$ we can make $4\left (\frac{2\gamma + (\lambda-\lambda_1)}{\lambda-\lambda_1} \right )^{3/2}$ progress, as required by Lemma \ref{acceleration_primitive} in time $O\left (\left (\nnz(\bv A) + dm \right) \cdot \log \left (\frac{\gamma}{\lambda-\lambda_1}\right) \right )$ time. Hence by Lemma \ref{acceleration_primitive} we can make constant factor expected progress in minimizing $f(x)$ in time:
\begin{align*}
O\left ( \left (\nnz(\bv A) + d\frac{\gamma^2 + 12\lambda_1\norm{\bv A}_F^2}{(\lambda-\lambda_1+\gamma)^2} \right ) \log \left (\frac{\gamma}{\lambda-\lambda_1}\right) \sqrt {\frac{\gamma}{\lambda-\lambda_1}}\right)
\end{align*}

By our assumption, we have $\nnz(\bv{A}) \le \frac{d\nrank(\bv A)}{\gap^2} = \frac{d \lambda_1 \norm{\bv A}_F^2}{(\lambda-\lambda_1)^2}$. So, if we let $\gamma = \Theta \left (\sqrt{\frac{d\lambda_1 \norm{\ma}_F^2}{\nnz(\bv A)}}\right)$ then using a sufficiently large constant, we have $\gamma \ge 2(\lambda-\lambda_1)$. We have $\frac{\gamma}{\lambda-\lambda_1} = \Theta \left (\sqrt{\frac{d\lambda_1 \norm{\bv A}_F^2}{\nnz(\bv A)\lambda_1^2 \gap^2}}\right) = \Theta \left (\sqrt{\frac{d\nrank(\bv A)}{\nnz(\bv A)\gap^2}}\right) $ and can make constant expected progress in minimizing $f(x)$ in time:
\begin{align*}
O\left ( \frac{\nnz(\bv A)^{3/4} (d \nrank(\bv A))^{1/4}}{\sqrt{\gap}} \cdot \log \left (\frac{d}{\gap}\right)\right).
\end{align*}
\end{proof}

\subsection{Shifted-and-Inverted Power Method}
\label{sec:offline:results}

Finally, we are able to combine the solvers from Sections \ref{sec:offline:svrg} and \ref{sec:offline:acceleration} with the framework of Section \ref{framework} to obtain faster algorithms for top eigenvector computation.
\begin{theorem}[Shifted-and-Inverted Power Method With SVRG]\label{main_offline_theorem}
Let $\bv{B} = \lambda \bv{I} - \bv{A}^\top \bv{A}$ for $\left ( 1+\frac{\gap}{150}\right) \lambda_1 \le \lambda \le \left(1+ \frac{\gap}{100} \right ) \lambda_1$ and let $x_0 \sim \mathcal{N}(0,\bv I)$ be a random initial vector. Running the inverted power method on $\bv{B}$ initialized with $x_0$, using the SVRG solver from Theorem \ref{offline_solver} to approximately apply $\bv{B}^{-1}$ at each step, returns $x$ such that with probability $1-O\left (\frac{1}{d^{10}}\right)$, $x^\top \bv{\Sigma}x \ge (1-\epsilon) \lambda_1$ in total time $$O \left (\left(\nnz(\bv A) + \frac{d \nrank(\bv A)}{\gap^2} \right )\cdot \left (\log^2\left(\frac{d}{\gap}\right) + \log\left(\frac{1}{\epsilon}\right) \right ) \right ).$$
\end{theorem}
Note that by instantiating the above theorem with $\epsilon' = \epsilon\cdot \gap$, and applying Lemma \ref{lem:ray_to_evec} we can find $x$ such that $| v_1^\top x | \ge 1-\epsilon$ in the same asymptotic running time (an extra $\log(1/\gap)$ term is absorbed into the $\log^2(d/\gap)$ term).
\begin{proof}
By Theorem \ref{thm:init-offline}, if we start with $x_0 \sim \mathcal{N}(0,\bv I)$ we can run $O\left (\log \left (\frac{d}{\gap} \right ) \right )$ iterations of the inverted power method, to obtain $x_1$ with $G(x_1) \le \frac{1}{\sqrt{10}}$ with probability $1-O\left (\frac{1}{d^{10}}\right)$. Each iteration requires applying an linear solver that decreases initial error in expectation by a factor of $\frac{1}{\poly(d,1/\gap)}$. Such a solver is given by applying the solver in Theorem \ref{offline_solver} $O\left (\log \left (\frac{d}{\gap} \right ) \right )$ times, decreasing error by a constant factor in expectation each time. So overall in order to find $x_1$ with $G(x_1) \le \frac{1}{\sqrt{10}}$, we require time $O \left (\left(\nnz(\bv A) + \frac{d \nrank(\bv A)}{\gap^2} \right )\cdot\log^2\left(\frac{d}{\gap}\right)\right )$.

After this initial `burn-in' period we can apply Corollary \ref{cor:constant_factor_corollary} of Theorem \ref{thm:powermethod-perturb}, which shows that running a single iteration of the inverted power method will decrease $G(x)$ by a constant factor in expectation. In such an iteration, we only need to use a solver that decreases initial error by a constant factor in expectation. So we can perform each inverted power iteration in this stage in time $O \left (\nnz(\bv A) + \frac{d \nrank(\bv A)}{\gap^2} \right ).$

With $O\left (\log \left (\frac{d}{\epsilon}\right)\right)$ iterations, we can obtain $x$ with $\E \left [G(x)^2 \right ] = O\left (\frac{\epsilon}{d^{10}}\right)$ So by Markov's inequality, we have $G(x)^2 = O(\epsilon)$, giving us $x^T \bv{\Sigma} x \ge (1-O(\epsilon))\lambda_1$ by Lemma \ref{lem:rayquot-potential}. Union bounding over both stages gives us failure probability $O\left ( \frac{1}{d^{10}}\right )$, and adding the runtimes from the two stages gives us the final result. Note that the second stage requires $O\left (\log \left (\frac{d}{\epsilon}\right)\right) = O(\log d + \log (1/\epsilon))$ iterations to achieve the high probability bound. However, the $O(\log d)$ term is smaller than the $O\left (\log^2 \left (\frac{d}{\gap} \right ) \right )$ term, so is absorbed into the asymptotic notation.

\end{proof}

We can apply an identical analysis using the accelerated solver from Theorem \ref{accelerated_offline_solver}, obtaining the following runtime which beats Theorem \ref{main_offline_theorem} whenever $\nnz(\bv A) \le \frac{d\nrank(\bv A)}{\gap^2}$:
\begin{theorem}[Shifted-and-Inverted Power Method Using Accelerated SVRG]\label{accelerated_offline_theorem}
Let $\bv{B} = \lambda \bv{I} - \bv{A}^\top \bv{A}$ for $\left ( 1+\frac{\gap}{150}\right) \lambda_1 \le \lambda \le \left(1+ \frac{\gap}{100} \right ) \lambda_1$ and let $x_0 \sim \mathcal{N}(0,\bv I)$ be a random initial vector. Assume that $\nnz(\bv A) \le \frac{d\nrank(\bv A)}{\gap^2}$. Running the inverted power method on $\bv{B}$ initialized with $x_0$, using the accelerated SVRG solver from Theorem \ref{accelerated_offline_solver} to approximately apply $\bv{B}^{-1}$ at each step, returns $x$ such that with probability $1-O\left (\frac{1}{d^{10}}\right)$, $| v_1^\top x | \ge 1-\epsilon$ in total time $$O \left (\left(\frac{\nnz(\bv A)^{3/4}(d \nrank(\bv A))^{1/4}}{\sqrt{\gap}} \right )\cdot \left (\log^3\left(\frac{d}{\gap}\right) + \log\left(\frac{d}{\gap}\right)\log\left(\frac{1}{\epsilon}\right) \right ) \right ).$$
\end{theorem}


%% file: online.tex
\section{Statistical Setting}\label{sec:online}

Here we show how to apply the shifted-and-inverted power method framework of Section \ref{framework} to the online setting. This setting is more difficult than the offline case. As there is no canonical matrix $\bv{A}$, and we only have access to the distribution $\dist$ through samples, in order to apply Theorem \ref{thm:powermethod-perturb}
we must show
 how to both estimate the Rayleigh quotient (Section~\ref{sec:online:rayleigh_quotient}) as well as solve the requisite linear systems in expectation (Section~\ref{sec:online:system-solver}).

After laying this ground work, our main result is  given in Section \ref{sec:online:result}. 
Ultimately, the results in this section allow us to achieve more efficient algorithms for computing the top eigenvector in the statistical setting as well as improve upon the previous best known sample complexity for top eigenvector computation. As we show in Section~\ref{sec:lower} the bounds we provide in this section are in fact tight for general distributions.

\subsection{Estimating the Rayleigh Quotient}
\label{sec:online:rayleigh_quotient}

Here we show how to estimate the Rayleigh quotient of a vector with respect to $\mSigma$. Our analysis is standard -- we first approximate the Rayleigh quotient by its empirical value on a batch of $k$ samples and prove using Chebyshev's inequality that the error on this sample is small with constant probability. We then repeat this procedure $O(\log (1/p))$ times and output the median. By Chernoff bound this yields a good estimate with probability $1 - p$. The formal statement of this result and its proof comprise the remainder of this subsection. 

\begin{theorem}[Online Rayleigh Quotient Estimation]\label{online_rayleigh_estimation}
Given any $\epsilon \in (0,1]$ and $p \in [0, 1]$ let $k = \lceil 4\nvar(\dist) \epsilon^{-2} \rceil$ and $m = c\log(1/p)$ for some sufficiently large constant $c$.
For all $i \in [k]$ and $j \in [m]$ let $a_{i}^{(j)}$ be drawn independently from $\dist$. Then if for any unit norm $x$ we let $R_{i,j} \defeq x^\top a_i^{(j)} (a_i^{(j)})^\top x$, $R_{j} \defeq \frac{1}{k} \sum_{i \in [k]} R_{i,j}$, and let $z$ be median value of the $R_j$ then with probability $1 - p$ we have that
$
\left|z - x^\top \mSigma x\right| \leq \epsilon \lambda_1 
$.
\end{theorem}

\begin{proof}
\begin{align*}
\variance_{a \sim \dist} (x^\top a a^\top x)
&=
\E_{a \sim \dist} (x^\top a a^\top x)^2 - (\E_{a \sim \dist} x^\top a a^\top x)^2
\\
&\leq 
\E_{a \sim \dist} \norm{a}_2^2 x^\top a a^\top x - (x^\top \mSigma x)^2
\\
&\leq \normFull{\E_{a \sim \dist} \norm{a}_2^2 a a^\top}_2 =\nvar(\dist) \lambda_1^2
\end{align*}
Consequently, $\variance(R_{i,j}) \leq \nvar(\dist) \lambda_1^2$, and since each of the $a_{i}^{(j)}$ were drawn independently this implies that we have that $\variance(R_j) \leq \nvar(\dist) \lambda_1^2 / k$. Therefore, by Chebyshev's inequality 
\[
\Pr\left[\left|R_j - \E[R_j]\right| \geq 2 \sqrt{\frac{\nvar(\dist) \lambda_1^2}{k}}\right]
\leq \frac{1}{4}.
\]
Since $\E[R_j] = x^\top \mSigma x$ and since we defined $k$ appropriately this implies that
\begin{equation}
\label{eq:rayleighlemma:1}
\Pr\left[\left|R_j - x^\top \mSigma x \right| \geq \epsilon \lambda_1 \right]
\leq \frac{1}{4}.
\end{equation}
The median $z$ satisfies $|z - x^\top \mSigma x| \leq \epsilon$ as more than half of the $R_j$ satisfy $|R_j - x^\top \mSigma x| \leq \epsilon$. This happens with probability $1 - p$ by Chernoff bound, our choice of $m$ and \eqref{eq:rayleighlemma:1}.
\end{proof}

\subsection{Solving the Linear system}
\label{sec:online:system-solver}

Here we show how to solve linear systems in $\mb = \lambda \mI - \mSigma$ in the streaming setting. We follow the general strategy of the offline algorithms given in Section \ref{sec:offline}, replacing traditional SVRG with the the streaming SVRG algorithm of \cite{frostig2014competing}. Again we minimize $f(x) = \frac{1}{2} x^\top \bv B x - b^\top x$.  Similarly to in the offline case, we define for all $a \in \supp(\dist)$,
\begin{align}\label{distribution_obj_function}
\psi_{a}(x)
	\defeq
\frac{1}{2} x^{\top} (\lambda \mI - aa^{\top})x - b^{\top}x.
\end{align}

This definition insures that
\[
f(x) = \E_{a \sim \dist} \psi_{a}(x).
\]

The performance of streaming SVRG \cite{frostig2014competing} will be governed by three regularity parameters. As in the offline case, we use the fact that $f(\cdot)$ is $\mu$-strongly convexity for $\mu = \lambda - \lambda_1$. We also again require a smoothness parameter $\avgsmooth$ such that:
\begin{equation}
\label{eq:smoothness}
\forall x \in \R^d \enspace : \enspace
\E_{a \sim \dist}\norm{\grad\psi_{a}(x)-\grad\psi_{a}(\opt{x})}_2^2
\leq 2 \avgsmooth \left [f(x) - f(\opt{x})\right].
\end{equation}

Lastly, we need an upper bound the variance, $\sigma^2$ such that:
\begin{equation}
\label{eq:variance}
\E_{a \sim \dist}    \frac{1}{2}
\normFull{\grad \psi_{a} (\opt{x})}_{
	\left(\hess f (\opt{x})\right)^{-1}}^2
\leq \sigma^2
\end{equation}

We bound the second two parameters as follows.

\begin{lemma} [Streaming Smoothness]
\label{lem:online:smooth}
The smoothness parameter
\[
\avgsmooth \defeq \lambda + \frac{\nvar(\dist)\lambda_1^2}{\lambda - \lambda_1}
\]
satisfies \eqref{eq:smoothness}.
\end{lemma}

\begin{proof} Our proof is similar to the one given in Lemma \ref{simple_variance_bound}.
\begin{align*}
\E_{a \sim \dist}
\normFull{\grad\psi_{a}(x) - \grad \psi_{a} (\opt{x})}_2^2 &= \E_{a \sim \dist} \normFull{(\lambda \mI - aa^\top) (x - \opt{x})}_2^2\\
	&=
	\lambda^2 \norm{x - \opt{x}}_2^2 - 2\lambda \E_{a \sim \dist} \norm{x - \opt{x}}_{aa^\top}^2 + \E_{a \sim \dist} \norm{aa^\top (x - \opt{x})}_2^2\\
	&\le
	\lambda^2 \norm{x - \opt{x}}_2^2 - 2\lambda \norm{x - \opt{x}}_{\bv \Sigma}^2 + \normFull{\E_{a \sim \dist} \norm{a}_2^2 a a^\top}_2 \cdot \norm{x - \opt{x}}_2^2\\
	&\leq  \lambda \normFull{x - \opt{x}}_\mb^2 + \nvar(\dist)\lambda_1^2\norm{x - \opt{x}}_2^2.
\end{align*}
Since $f$ is $\lambda - \lambda_1$-strongly convex, $\norm{x - \opt{x}}_2^2 \leq \frac{2}{\lambda - \lambda_1} [f(x) - f(\opt{x})]$. Additionally, $2[f(x) - f(\opt{x})] = \norm{x - \opt{x}}_{\mb}^2$. The result follows.
\end{proof}

\begin{lemma}[Streaming Variance] 
\label{lem:online:streamvar}
For $
\sigma^2
\defeq
\E_{a \sim \dist}  \frac{1}{2}
\normFull{\grad \psi_{a} (\opt{x})}_{
	\left(\hess f (\opt{x})\right)^{-1}}^2 $
we have
\[
\sigma^{2} 
\leq 
\left(\frac{\nvar(\dist)\lambda_1^2}{\lambda - \lambda_1} \right)
\normFull{\opt{x}}_2^2.
\]
\end{lemma}

\begin{proof}
We have
\begin{align*}
\sigma^2
	 &=
\E_{a \sim \dist} \frac{1}{2} \normFull{\left(\lambda \mI 
	- aa^{\top}\right) \opt{x} - b}_{\mb^{-1}}^{2}
	\\
	&=
\E_{a \sim \dist} \frac{1}{2} \normFull{\left(\lambda \mI 
	- aa^{\top}\right) \opt{x} - \mb \opt{x}}_{\mb^{-1}}^{2}
	\\
	&=
\E_{a \sim \dist} \frac{1}{2} \normFull{\left(\mSigma
	- aa^{\top}\right) \opt{x}}_{\mb^{-1}}^{2}.
\end{align*}
Applying $\E \norm{a -\E a}_{2}^{2} = \E\norm a_{2}^{2}-\norm{\E a}_{2}^{2}$ gives:
\[
\sigma^2 = \E_{a \sim \dist} \frac{1}{2} \normFull{\opt{x}}_{a a^\top \mb^{-1} a a^\top}^2
- \frac{1}{2} \normFull{\opt{x}}_{\mSigma \mb^{-1} \mSigma}^2 \le \E_{a \sim \dist} \frac{1}{2} \normFull{\opt{x}}_{a a^\top \mb^{-1} a a^\top}^2.
\]
Furthermore, since $\mb^{-1} \preceq \frac{1}{\lambda - \lambda_1} \mI$ we have 
\[
\E_{a \sim \dist}  a a^{\top} \mb^{-1} a a^\top
\preceq
\frac{1}{\lambda - \lambda_1} \E_{a \sim \dist}  (a a^\top)^2
\preceq
\left(
\frac{\normFull{\E_{a \sim \dist} (a a ^\top)^2}_2}{\lambda - \lambda_1}\right) \mI
=
\left(\frac{\nvar(\dist)\lambda_1^2}{\lambda - \lambda_1}\right) \mI.
\]
yielding the result.
\end{proof}

With these two emmas in place, we can apply the streaming SVRG algorithm of \cite{frostig2014competing} to solving systems in $\bv{B}$. We encapsulate the basic iterative step of  Algorithm $1$ of \cite{frostig2014competing} in the following definition:


\begin{definition}[Streaming SVRG Step]\label{svrg_step_def} 
Given $x_0 \in \R^d$ and $\eta, k, m > 0$ we define a \emph{streaming SVRG step}, $x = \ssvrgstep(x_0, \eta, k, m)$ as follows. First we take $k$ samples $a_1, ..., a_k$ from $\dist$ and set $g = \frac{1}{k} \sum_{i \in [k]} \psi_{a_i}$ where $\psi_{a_i}$ is as defined in \eqref{distribution_obj_function}. Then for $\wt{m}$ chosen uniformly at random from $\{1, ..., m\}$ we draw $\tilde m$ additional samples $\wt{a}_1, ..., \wt{a}_{\wt{m}}$ from $\dist$. For $t = 0, ... , \wt{m} - 1$ we let
\[
x_{t+1}
	:=
x_t -
	\frac{\eta}{L}
\left(
	\grad \psi_{\wt{a}_{t}}(x_t)
	- \grad \psi_{\wt{a}_{t}}(x_0)
	+ \grad g(x_0)
\right)
	\]
	Finally, we return $x_{\wt{m}}$.
\end{definition}

The accuracy of the above iterative step is proven in Theorem 4.1 of \cite{frostig2014competing}, which we include, using our notation below:

\begin{theorem} [Theorem 4.1 of \cite{frostig2014competing}]\label{streaming_svrg_perf_bound}
Letting $f(x) = \E_{a\sim \dist} \psi_a(x)$ and $\mu$, $\avgsmooth$, $\sigma^2$ be the strong convexity, smoothness, and variance bounds for $f(x)$, for any distribution over $x_0$ we have that 
$x := \ssvrgstep(w_0, \eta, k, m)$ has
$\E[f(x) - f(\opt{x})]$ upper bounded by
\[
	\frac{1}{1 - 4\eta}
\left[
	\left(\frac{\avgsmooth}{\mu m\eta} + 4 \eta \right)
	\left [f(x_0) - f(\opt{x}) \right ]
+ \frac{1 + 2\eta}{k}
	\left(
	\sqrt{\frac{\avgsmooth}{\mu} \cdot \left [f(x_0) - f(\opt{x})\right ]}
	+ \sigma
	\right)^2
	\right].
\]
\end{theorem}
Note that 
Theorem~4.1 in \cite{frostig2014competing} has an additional parameter of $\alpha$, which bounds the Hessian of $f(\cdot )$ at the optimum in comparison to the hessian at at any other point. In our setting this parameter is $1$ as $\hess f(y) = \hess f(z)$ for all $y$ and $z$.
With Theorem \ref{streaming_svrg_perf_bound} we immediately have the following:

\begin{corollary}[Streaming SVRG-Based Solver]\label{streaming_solverOLD}
Let $\mu =\lambda-\lambda_1$, $\avgsmooth =\lambda + \frac{\nvar(\dist)\lambda_1^2}{\lambda-\lambda_1}$, and $\sigma^2 = \frac{\nvar(\dist)\lambda_1^2}{\lambda-\lambda_1} \norm{\opt{x}}_2^2$. Let $c_2,c_3 \in (0,1)$ be any constants, $\eta = \frac{c_2}{8}$, $m = \left [\frac{\avgsmooth}{\mu c_2^2}\right]$, and $k = \max \left \{ \left [\frac{\avgsmooth}{\mu c_2}\right], \left [\frac{\nvar(\dist)\lambda_1^2}{(\lambda-\lambda_1)^2c_3} \right] \right \}$. If used to solve $\bv{B}x = b$ for a unit vector $b$, the iterative procedure described in Definition \ref{svrg_step_def} returns $x = \ssvrgstep(x_0, \eta, k, m)$ satisfying:
\[
\E \norm{x - \opt{x}}_{\mb}^2
\leq 22c_2 \cdot \norm{x_0-\opt{x}}_\mb^2 + 10 c_3\lambda_1(\bv{B}^{-1}).
\]
Further, the procedure requires $O \left ( \frac{\nvar(\dist)}{\gap^2} \left [\frac{1}{c_2^2} + \frac{1}{c_3} \right ] \right )$ samples from $\dist$.
\end{corollary}

\begin{proof}
Using the inequality $(x + y)^2 \leq 2x^2 + 2y^2$ we have that
\[
\left(\sqrt{\frac{\avgsmooth}{\mu} \cdot \E[f(x_0) - f(\opt{x})]} + \sigma \right)^2
\leq \frac{2 \avgsmooth}{\mu} \cdot \E[f(x_0) - f(\opt{x})]
+ 2\sigma^2
\]
Additionally, since $b$ is a unit vector, we know that $\norm{\opt{x}}_2^2 = \norm{\mb^{-1} b}_2^2 \le \frac{1}{(\lambda-\lambda_1)^2}$.
Using the fact shown in equation \eqref{norm_to_function_error_conversion} that for any $x$, $\norm{x-\opt{x}}_\bv{B}^2 = 2 [f(x)-f(\opt{x})]$ we have by Theorem \ref{streaming_svrg_perf_bound}:
\begin{align*}
\E \norm{x-\opt{x}}_\bv{B}^2 &\le \frac{1}{1-c_2/2} \left [8c_2 + \frac{c_2}{2} + \frac{4+c_2}{2}\cdot c_2 \right ]\cdot \norm{x_0-\opt{x}}_\mb^2\\
&+ \frac{4}{1-c_2/2}\cdot \frac{4+c_2}{4k}\cdot \frac{\nvar(\dist)\lambda_1^2}{(\lambda-\lambda_1)^3}\\
&\le 22c \cdot \norm{x_0-\opt{x}}_\mb^2 + \frac{10 c_3}{\lambda-\lambda_1}\\
& = 22c \cdot \norm{x_0-\opt{x}}_\mb^2 + 10 c_3 \lambda_1(\bv{B}^{-1}).
\end{align*}

The number of samples required to make the streaming SVRG step is simply bounded by $m+k$. $m = \frac{\avgsmooth}{\mu c_2^2} = O  \left(\frac{\lambda}{c_2^2(\lambda-\lambda_1)} + \frac{\nvar(\dist)\lambda_1^2}{c_2^2(\lambda-\lambda_1)^2}\right ) = O\left (\frac{1}{c_2^2\gap} + \frac{\nvar(\dist)}{c_2^2\gap^2} \right )$. $\gap < 1$ and $\nvar(\dist) \ge 1$.  So we can simplify: $m = O\left ( \frac{\nvar}{c_2^2\gap^2}\right )$. We can ignore the $\left [\frac{\avgsmooth}{\mu c_2}\right]$ in $k$ since this was already included in our bound of $m$ and so just bound $\frac{\nvar(\dist)\lambda_1^2}{c_3(\lambda-\lambda_1)^2} = O\left ( \frac{\nvar(\dist)}{\gap^2c_3}\right )$. So overall the number of samples we need to take is:
\begin{align*}
O \left ( \frac{\nvar(\dist)}{\gap^2} \left [\frac{1}{c_2^2} + \frac{1}{c_3} \right ] \right ).
\end{align*}
\end{proof}

In the offline case, when solving linear systems in the shifted-and-inverted power method, we can insure that $\norm{x_0 - \opt{x}}_\mb^2$ is small by Corollary \ref{cor:constant_factor_corollary}. In the online case, we do not have the same guarantee, since starting with a good initial value requires accurately estimating our Rayleigh quotient, which is too expensive. However, we can still show the following corollary:

\begin{corollary}[Streaming SVRG-Based Solver With No Initial Error] \label{streaming_solver}
There is a streaming algorithm that iteratively applies the solver of Corollary \ref{streaming_solverOLD} to
$\bv{B}x = b$ for a unit vector $b$, which returns $x$ satisfying:
\[
\E \norm{x - \opt{x}}_{\mb}^2
\leq 10 c_3\lambda_1(\bv{B}^{-1}).
\]
using $O \left ( \frac{\nvar(\dist)}{\gap^2c_3}\right )$ samples from $\dist$.
\end{corollary}
\begin{proof}
Let $x_0 = 0$. Then $\norm{x_0-\opt{x}}_\mb^2 = \norm{\bv{B}^{-1}x}_\mb^2 \le \lambda_1(\bv{B}^{-1})$ since $x$ is a unit vector. If we apply Corollary \ref{streaming_solverOLD} with $c_2 = 1/44$ and $c_3' = \frac{1}{20}$, then we will obtain $x_1$ with $\E \norm{x_1-\opt{x}}_\mb^2 \le \frac{1}{2}\lambda_1(\mb^{-1})$. If we then double $c_3'$ and apply the solver again we obtain $x_2$ with  $\E \norm{x_1-\opt{x}}_\mb^2 \le \frac{1}{4}\lambda_1(\mb^{-1})$. Iterating in this way, after $\log(1/c_3)$ iterations we will have the desired guarantee: $\E \norm{x - \opt{x}}_{\mb}^2 \leq 10 c_3\lambda_1(\bv{B}^{-1}).$ Our total sample cost in each iteration is, by Corollary \ref{streaming_solverOLD}, $O \left ( \frac{\nvar(\dist)}{\gap^2} \left [\frac{1}{44^2} + \frac{1}{c_3'} \right ] \right ).$ Since we double $c_3'$ each time, the cost corresponding to the $\frac{1}{c_3'}$ terms is dominated by the last iteration when we have $c_3' = O(c_3)$. So our overall sample cost is just:
\begin{align*}
O \left ( \frac{\nvar(\dist)}{\gap^2}\left [\frac{1}{c_3} + \log(1/c_3) \right ] \right ) = O \left ( \frac{\nvar(\dist)}{\gap^2c_3}\right ).
\end{align*}
\end{proof}

\subsection{Online Shifted-and-Inverted Power Method}
\label{sec:online:result}

We now apply the results in Section~\ref{sec:online:rayleigh_quotient} and Section~\ref{sec:online:system-solver} to the power method framework of Section \ref{framework} to give our main online result. 
We give an algorithm that quickly refines a coarse approximation to $v_1$ into a finer approximation. 

\begin{theorem}[Online Shifted-and-Inverted Power Method -- Warm Start]\label{warmstart_online_theorem}
Let $\bv{B} = \lambda \bv{I} - \bv{A}^\top \bv{A}$ for $\left ( 1+\frac{\gap}{150}\right) \lambda_1 \le \lambda \le \left(1+ \frac{\gap}{100} \right ) \lambda_1$ and let $x_0$ be some vector with  $G(x_0)\leq \frac{1}{\sqrt{10}}$. Running the inverted power method on $\bv{B}$ initialized with $x_0$, using the streaming SVRG solver described in Definition \ref{svrg_step_def} to approximately apply $\bv{B}^{-1}$ at each step, returns $x$ such that, for any parameter $\delta > \frac{1}{d^{10}}$ with probability $1-\delta$, $x^\top \bv{\Sigma} x \ge (1-\epsilon) \lambda_1$ using total sample count:
\begin{align*}
O \left (\frac{\nvar(\dist)}{\gap} \cdot \left (\frac{\log 1/\delta+\log\log 1/\epsilon}{\gap} + \frac{1}{\delta^2\epsilon} \right) \right )
\end{align*}
The amortized processing time per simple is simple $O(d)$.
\end{theorem}
We note that by instantiating Theorem \ref{warmstart_online_theorem}, with $\epsilon' = \epsilon\cdot \gap$, and applying Lemma \ref{lem:ray_to_evec} we can find $x$ such that $| v_1^\top x | \ge 1-\epsilon$ in time $O \left (\frac{\nvar(\dist)}{\gap^2 \cdot \epsilon \cdot \delta^2} \right ).
$
\begin{proof}
By Lemma \ref{lem:rayquot-potential} it suffices to have $G^2(x) = O(\frac{\epsilon}{\gap})$ so $G(x) = O(\sqrt{\epsilon/\gap})$ . In order to have this with probability $1-\delta$ it suffices to have $\expec{ G(x)} = O(\delta \sqrt{\epsilon/\gap})$. Since we start with $G(x_0)\leq \frac{1}{\sqrt{10}}$, we can achieve this
using $\log(\gap/(\delta\epsilon))$ iterations of the approximate power method of Theorem \ref{thm:powermethod-perturb}. In iteration $i$ we choose the error parameter for Theorem \ref{thm:powermethod-perturb} to be $c_1(i) = \frac{1}{\sqrt{10}}\cdot \left( \frac{1}{5}\right)^i$. In this way, we have:
\begin{align*}
\expec{G(x_i)} \le \frac{3}{25} G(x_{i-1}) +  \frac{4}{1000}\frac{1}{\sqrt{10}}\cdot \left( \frac{1}{5}\right)^i
\end{align*}
and by induction $\expec{G(x_i)} \le \frac{1}{5^i} \frac{1}{\sqrt{10}}$. We halt when $(\frac{1}{5})^i = O(\delta \sqrt{\epsilon/\gap})$ and hence $c_1(i) = O(\delta \sqrt{\epsilon/\gap})$.

In order to apply Theorem \ref{thm:powermethod-perturb} we need a subroutine $\rayquoth{x}$ that lets us approximate $\rayquot(x)$ to within an additive error $\frac{1}{30} (\lambda-\lambda_1) = O(\gap\lambda_1)$. Theorem \ref{online_rayleigh_estimation} gives us such a routine, requiring $O \left (\frac{\nvar(\dist)\log (1/\delta\log(\gap/\delta\epsilon))}{\gap^2} \right)$ samples to succeed with probability $1- O\left(\frac{\delta}{\log(\gap/\delta\epsilon)}\right)$. Union bounding, the estimation succeeds in all rounds of the power method with probability at least $1-O(\delta)$.

Our cost for each linear system solve is given by Corollary \ref{streaming_solver} with $c_3 = \Theta(c_1(i)^2)$
, is:
\begin{align*}
O \left ( \frac{\nvar(\dist)}{\gap^2c_1(i)^2} \right ).
\end{align*}

Now, $c_1(i)$ multiplies by a constant factor with each iteration. The cost over all $O(\log(\gap/d\epsilon)$ iterations is just a truncated geometric series and is proportional to cost in the last iteration, when $c_3 = \Theta\left(\frac{\delta^2 \epsilon}{\gap}\right)$. So the total cost for solving the linear systems is 
\begin{align*}
O \left ( \frac{\nvar(\dist)}{\gap \delta^2\epsilon} \right ).
\end{align*}

Adding this to the Rayleigh quotient estimation cost give us total sample count:
\begin{align*}
O \left (\frac{\nvar(\dist)}{\gap^2} \cdot \log (1/\delta\log(1/\epsilon)) + \frac{\nvar(\dist)}{\gap\delta^2\epsilon} \right ) = O \left (\frac{\nvar(\dist)}{\gap} \cdot \left (\frac{\log 1/\delta+\log\log 1/\epsilon}{\gap} + \frac{1}{\delta^2\epsilon} \right) \right ) .
\end{align*}
\end{proof}

%% file: paramfree.tex
\section{Parameter Estimation for Offline Eigenvector Computation}\label{parameter_free}


In Section \ref{sec:offline}, in order to invoke Theorems \ref{thm:powermethod-perturb} and \ref{thm:init-offline} we assumed knowledge of some $\lambda$ with $(1 + c_1 \cdot \gap)\lambda_1 \le \lambda \le (1 + c_2 \cdot \gap)\lambda_1$ for some small constant $c_1$ and $c_2$ ($\hat{\lambda}_1$ in Theorem \ref{thm:powermethod-perturb} could also be obtained in this form). Here we show how to estimate this parameter using Algorithm \ref{algo:eigestimate}, completing the proof of our offline eigenvector estimation algorithm (Theorem \ref{main_offline_theorem}).

In this section, for simplicity we assume that we have oracle access to compute $\mb_{\lambda}^{-1}w$ for any given $w$, and any $\lambda > \lambda_1$, but the results here can be extended to the case where we can compute $\mb_{\lambda}^{-1}w$ only approximately. We will use a result of~\cite{Musco2015} that gives gap free bounds for computing eigenvalues using the power method. The following is a specialization of Theorem 1 from~\cite{Musco2015}:
\begin{theorem}\label{thm:musco}
	For any $\epsilon > 0$, any matrix $\mM \in \R^{d\times d}$ with eigenvalues $\lambda_1,...,\lambda_d$, and $k \leq d$, let $\bv{W} \in \mathbb{R}^{d \times k}$ be a matrix with entries drawn independently from $\mathcal{N}(0,1)$. Let $eigEstimate (\bv{Y})$ be a function returning for each $i$, $\tilde \lambda_i = \tilde v_i^\top \bv{M} \tilde v_i$ where $\tilde v_i$ is the $i^{th}$ largest eigenvector of  $\bv{Y}$. Then setting $[\tilde \lambda_1,...,\tilde \lambda_k ] = eigEstimate\left (\bv{M}^t \bv{W} \right )$, for some fixed constant $c$ and $t = c\alpha \log d$ for any $\alpha > 1$, we have for all $i$:
	\begin{align*}
	|\tilde \lambda_i - \lambda_i|  \le \frac{1}{\alpha}\lambda_{k+1}
	\end{align*}
\end{theorem}

\begin{algorithm}[t]
	\caption{Estimating the eigenvalue and the eigengap}
	\begin{algorithmic}[1]
		\renewcommand{\algorithmicrequire}{\textbf{Input: }}
		\renewcommand{\algorithmicensure}{\textbf{Output: }}
		\REQUIRE $\bv{A}\in\mathbb{R}^{n\times d},\;\alpha$
		\STATE $\left[w_{1},w_{2}\right]\leftarrow\mathcal{N}\left(0,I^{d\times d}\right)$
		\STATE $t\leftarrow O\left(\alpha\log d\right)$
		\STATE $\left[\lamtilij 01,\lamtilij 02\right]\leftarrow eigEstimate\left(\left(\bv{A}^{T}\bv{A}\right)^{t}w\right)$
		\STATE $\lambari 0\leftarrow(1+\frac{1}{2})\lamtilij 01$
		\STATE $i\leftarrow0$
		\WHILE{$\lambari i-\lamtilij i1<\frac{1}{10}\left(\lambari i-\lamtilij i2\right)$}
		\STATE $i\leftarrow i+1$
		\STATE $\left[w_{1},w_{2}\right]\leftarrow\mathcal{N}\left(0,I^{d\times d}\right)$
		\STATE $\left[\lamhatij i1,\lamhatij i2\right]\leftarrow eigEstimate\left(\left(\lambari{i-1}\bv{I}-\bv{A}^{T}\bv{A}\right)^{-t}w\right)$
		\STATE $\left[\lamtilij i1,\lamtilij i2\right]\leftarrow\left[\lambari{i-1}-\frac{1}{\lamhatij i1},\lambari{i-1}-\frac{1}{\lamhatij i2}\right]$
		\STATE $\lambari i\leftarrow\frac{1}{2}\left(\lamtilij i1+\lambari{i-1}\right)$
		\ENDWHILE
		\ENSURE $\lambda$
	\end{algorithmic}
	\label{algo:eigestimate}
\end{algorithm}

Throughout the proof, we assume $\alpha$ is picked to be some large constant so that $\alpha>100$, then Theorem \ref{thm:musco} implies:
\begin{lemma}
	\label{lem:topeig1} The iterates
	of Algorithm \ref{algo:eigestimate} satisfy:
	\begin{align*}
	0\leq\lamj 1-\lamtilij 01\leq\frac{1}{\alpha}\lamj 1\;\;\mbox{and}\;\; & \frac{1}{2}\left(1-\frac{3}{\alpha}\right)\lamj 1\leq\lambari 0-\lamj 1\leq\frac{1}{2}\lamj 1,\;\;\mbox{and,}
	\end{align*}
	\begin{align*}
	0\leq\lamj 1-\lamtilij i1\leq\frac{1}{\alpha-1}\left(\lambari{i-1}-\lamj 1\right)\;\;\mbox{and}\;\; & \frac{1}{2}\left(1-\frac{1}{\alpha-1}\right)\left(\lambari{i-1}-\lamj 1\right)\leq\lambari i-\lamj 1\leq\frac{1}{2}\left(\lambari{i-1}-\lamj 1\right).
	\end{align*}
\end{lemma}
\begin{proof}
	The proof can be decomposed into two parts:
	
	\textbf{Part I (Lines 3-4):} Theorem \ref{thm:musco} tells us that $\lamtilij 01\geq\left(1-\frac{1}{\alpha}\right)\lamj 1$.
	This means that we have
	\begin{align*}
	0\leq\lamj 1-\lamtilij 01\leq\frac{1}{\alpha}\lamj 1\;\;\mbox{and}\;\; & \frac{1}{2}\left(1-\frac{3}{\alpha}\right)\lamj 1\leq\lambari 0-\lamj 1\leq\frac{1}{2}\lamj 1.
	\end{align*}

	\textbf{Part II (Lines 5-6):} Consider now iteration $i$. We now
	apply Theorem \ref{thm:musco} to the matrix $\left(\lambari{i-1}\bv{I}-\bv{A}^{T}\bv{A}\right)^{-1}$.
	The top eigenvalue of this matrix is $\left(\lambari{i-1}-\lamj 1\right)^{-1}$.
	This means that we have $\left(1-\frac{1}{\alpha}\right)\left(\lambari{i-1}-\lamj 1\right)^{-1}\leq\lamhatij i1\leq\left(\lambari{i-1}-\lamj 1\right)^{-1}$,
	and hence we have, 
	\begin{align*}
	0\leq\lamj 1-\lamtilij i1\leq\frac{1}{\alpha-1}\left(\lambari{i-1}-\lamj 1\right)\;\;\mbox{and}\;\; & \frac{1}{2}\left(1-\frac{1}{\alpha-1}\right)\left(\lambari{i-1}-\lamj 1\right)\leq\lambari i-\lamj 1\leq\frac{1}{2}\left(\lambari{i-1}-\lamj 1\right).
	\end{align*}

	This proves the lemma.\end{proof}
\begin{lemma}
	\label{lem:secondeig}Recall we denote $\lamj 2\defeq\lamj 2\left(\bv{A}^{T}\bv{A}\right)$
	and $gap\defeq\frac{\lamj 1-\lamj 2}{\lamj 1}$. The iterates
	of Algorithm \ref{algo:eigestimate} satisfy $\abs{\lamj 2-\lamtilij i2}\leq\frac{1}{\alpha-1}\left(\lambari{i-1}-\lamj 2\right)$,
	and $\lambari i-\lamtilij i2\geq\frac{gap \lamj 1}{4}$.\end{lemma}
\begin{proof}
	Since $\left(\lambari{i-1}-\lamj 2\right)^{-1}$ is the second eigenvalue
	of the matrix $\left(\lambari{i-1}\bv{I}-\bv{A}^{T}\bv{A}\right)^{-1}$, Theorem \ref{thm:musco} 
	tells us that
	\[
	\left(1-\frac{1}{\alpha}\right)\left(\lambari{i-1}-\lamj 2\right)^{-1}\leq\lamhatij i2\leq\left(1+\frac{1}{\alpha}\right)\left(\lambari{i-1}-\lamj 2\right)^{-1}.
	\]

	This immediately yields the first claim. For the second claim, we
	notice that
	\begin{align*}
	\lambari i-\lamtilij i2 
	&=\lambari i-\lamj 2+\lamj 2-\lamtilij i2\\
	& \stackrel{(\zeta_{1})}{\geq}\lambari i-\lamj 2-\frac{1}{\alpha-1}\left(\lambari{i-1}-\lamj 2\right) \\
	& =\lambari i-\lamj 1-\frac{1}{\alpha-1}\left(\lambari{i-1}-\lamj 1\right)+\left(1-\frac{1}{\alpha-1}\right)\left(\lamj 1-\lamj 2\right) \\
	& \stackrel{\left(\zeta_{2}\right)}{\geq}\frac{1}{2}\left(1-\frac{3}{\alpha-1}\right)\left(\lambari{i-1}-\lamj 1\right)+\left(1-\frac{1}{\alpha-1}\right)\left(\lamj 1-\lamj 2\right)\geq\frac{gap\lamj 1}{4},
	\end{align*}

	where $\left(\zeta_{1}\right)$ follows from the first claim of this lemma,
	and $\left(\zeta_{2}\right)$ follows from Lemma \ref{lem:topeig1}.
\end{proof}

We now state and prove the final result:
\begin{theorem}\label{thm:parafree}
Suppose $\alpha>100$, and after $T$ iterations, Algorithm \ref{algo:eigestimate} exits
(i.e. $T$ is the final value of iterates $i$ in Algorithm), then we have
$T \le \ceil{\log\frac{10}{gap}}+1$, and we will have:
\begin{equation*}
\left(1+\frac{gap}{120}\right)\lambda_1 
\le \lambari T 
\le \left(1+\frac{gap}{8}\right)\lambda_1
\end{equation*}
\end{theorem}

\begin{proof}
Let $\overline{i}=\ceil{\log\frac{10}{gap}}$, suppose the algorithm has not exited yet after $\overline{i}$
iterations, then since $\lambari i-\lamj 1$ decays geometrically, we have $\lambari{\overline{i}}-\lamj 1\leq\frac{gap\lamj 1}{10}$.
Therefore, Lemmas \ref{lem:topeig1} and \ref{lem:secondeig} imply
that $\lambari{\overline{i}+1}-\lamtilij{\overline{i}+1}1\leq\left(\frac{1}{2}+\frac{1}{\alpha-1}\right)\left(\lambari{\overline{i}}-\lamj 1\right)\leq\frac{gap\lamj 1}{15}$,
and 
\begin{align*}
\lambari{\overline{i}+1}-\lamtilij{\overline{i}+1}2
&\geq\lambari{\overline{i}+1}-\lamj 2-\abs{\lamj 2-\lamtilij{\overline{i}+1}2}\geq\lamj 1-\lamj 2-\frac{1}{\alpha-1}\left(\lambari{\overline{i}}-\lamj 2\right)\\
&=gap\lamj 1-\frac{1}{\alpha-1}\left(\lambari{\overline{i}}-\lamj 1+\lamj 1-\lamj 2\right)\geq\frac{3}{4}gap\lamj 1
\end{align*}
This means that the exit condition on Line $6$ must be triggered in $\overline{i}+1$ iteration, proving
the first part of the lemma.

For upper bound, by Lemmas \ref{lem:topeig1}, \ref{lem:secondeig} and exit condition we know:
\begin{align*}
\lambari T- \lambda_1 &\le \lambari T-\lamtilij T1
\le \frac{1}{10} (\lambari T-\lamtilij T2)
\le \frac{1}{10} \left(\lambari T - \lambda_2 + \abs{\lamj 2-\lamtilij T2}\right)\\
&\leq \frac{1}{10} \left(\lambari T - \lambda_2 + \frac{1}{\alpha-1}(\lambari{T-1}-\lamj 2)\right)\\
&= \frac{1}{10} \left( \frac{\alpha}{\alpha-1} gap \lambda_1 + (\lambari T  - \lambda_1) +  \frac{1}{\alpha-1}\left(\lambari{T-1}-\lamj 1\right) \right)\\
&\le \frac{1}{10} \left( \frac{\alpha}{\alpha-1} gap \lambda_1 + \frac{\alpha}{\alpha-2}\left(\lambari{T}-\lamj 1\right) \right)
\end{align*}
Since $\alpha > 100$, this directly implies $\lambari{T}-\lamj 1 \le \frac{gap}{8}\lambda_1$.

For lower bound, since as long as the Algorithm \ref{algo:eigestimate} does not exists, by Lemmas \ref{lem:secondeig}, we have
$\lambari {T-1}-\lamtilij {T-1}1\geq\frac{1}{10}\left(\lambari {T-1}-\lamtilij {T-1}2\right)\geq\frac{gap\lamj 1}{40}$, and thus:
\begin{align*}
\lambari {T-1}- \lambda_1 &= \lambari {T-1}-\lamtilij {T-1}1 - (\lambda_1 - \lamtilij {T-1}1)
\ge \frac{gap\lamj 1}{40} - \frac{1}{\alpha-1}\left(\lambari{T-1}-\lamj 1\right) \\
& \ge \frac{gap\lamj 1}{40} - \frac{2}{\alpha-2}\left(\lambari{T}-\lamj 1\right)
\ge \frac{gap\lamj 1}{50}
\end{align*}
By Lemma \ref{lem:topeig1}, we know $\lambari {T}- \lambda_1\ge \frac{1}{2}(1-\frac{1}{\alpha-1}(\lambari {T-1}- \lambda_1)) >\frac{gap}{120}\lamj 1$
\end{proof}

Note that, although we proved the upper bound and lower bound in Theorem \ref{thm:parafree} with specific constants coefficient $\frac{1}{8}$ and $\frac{1}{120}$, this analysis can easily be extended  to any smaller constants by modifying the constant in exit condition, and choosing $\alpha$ larger.

%% file: lower.tex
\section{Lower Bounds}\label{sec:lower}

Here we show that our online eigenvector estimation algorithm (Theorem \ref{warmstart_online_theorem})  is asymptotically optimal - as sample size grows large it achieves optimal accuracy as a function of sample size. 
We rely on the following lower bound for eigenvector estimation in the Gaussian spike model:

\begin{lemma}[Lower bound for Gaussian Spike Model \cite{birnbaum2013minimax}] \label{lowerbound_spike_lemma}
Suppose data is generated as 
\begin{equation}\label{spike_model}
a_i = \sqrt{\lambda}\iota_i{v^\star} + Z_i
\end{equation}
where $\iota_i \sim \mathcal{N}(0, 1)$, and $Z_i \sim \mathcal{N}(0, I_d)$. 
Let $\hat{v}$ be some estimator of top eigenvector $v^\star$. 
Then, there exist some universal constant $c_0$, so that for $n$ sufficiently large, we have:
\begin{equation*}
\inf_{\hat{v}} \max_{v^\star \in \mathbb{S}^{d-1}} \E \norm{\hat{v} - v^\star}_2
\ge c_0 \frac{(1+\lambda)d}{\lambda^2 n}
\end{equation*}
\end{lemma}


\begin{theorem}\label{lowerbound_online_theorem}
Consider the problem of estimating the top eigenvector $v_1$ of $\E_{a\sim \mathcal{D}} aa^\top$, where we observe $n$ i.i.d samples from unknown distribution $\mathcal{D}$. If $\gap < 0.9$, then there exist some universal constant c, for any estimator $\hat{v}$ of top eigenvector, there always exists some hard distribution $\mathcal{D}$ so that for $n$ sufficiently large:
\begin{align*}
\E\norm{\hat{v}-v_1}^2_2 \ge c  \frac{\nvar(\dist)}{\gap^2n} 
\end{align*}
\end{theorem}

\begin{proof}
Suppose the claim of theorem is not true, then there exist some estimator $\hat{v}$ so that 
\begin{equation*}
\E\norm{\hat{v}-v_1}^2_2 < c'  \frac{\nvar(\dist)}{\gap^2n} 
\end{equation*}
holds for all distribution $\mathcal{D}$, and for any fixed constant $c'$ when $n$ is sufficiently large.

Let distribution $\mathcal{D}$ be the Gaussian Spike Model specified by Eq.(\ref{spike_model}), then
by calculation, it's not hard to verify that:
\begin{equation*}
\nvar(\dist) = \frac{\norm{\E_{a \sim \dist}\left [ \left (a a^\top \right )^2 \right ]}_2}{\norm{\E_{a \sim \dist} (aa^\top)}_2^2} =\frac{d+2+3\lambda}{1+\lambda}
\end{equation*}Since we know $\gap = \frac{\lambda}{1+\lambda} <0.9$, this implies $\lambda <9$, which gives
$\nvar(\dist) < \frac{d+29}{1+\lambda} <\frac{30d}{1+\lambda}$.
Therefore, we have:
\begin{equation*}
\E\norm{\hat{v}-v^\star}^2_2 < c'  \frac{\nvar(\dist)}{\gap^2n} <30c' \frac{(1+\lambda)d}{\lambda^2 n}
\end{equation*}
holds for all $v^\star \in \mathbb{S}^{d-1}$. Choose $c' = \frac{c_0}{30}$ in Lemma \ref{lowerbound_spike_lemma} we have a contradiction.
\end{proof}

$\norm{\hat{v}-v_1}^2_2  = 2 - 2  \hat{v}^\top v_1$, so this bound implies that- to obtain $|\hat{v}^\top v_1| \ge 1- \epsilon$,  we need $\frac{\nvar(\dist)}{\gap^2 n} = O(\epsilon)$ so $n = \Theta \left (\frac{\nvar(\dist)}{\gap^2 \epsilon} \right )$. This exactly matches the sample complexity given by Theorem \ref{warmstart_online_theorem}.